\documentclass[aps,twocolumn,superscriptaddress,floatfix,prb]{revtex4-1}

\usepackage{times}

\usepackage{amsmath}
\usepackage{graphicx}
\usepackage[utf8]{inputenc}
\usepackage{braket}
\usepackage{amssymb}
\usepackage[colorlinks=true,citecolor=blue,linkcolor=blue,urlcolor=blue]{hyperref}
\usepackage[english]{babel}
\usepackage[normalem]{ulem}
\usepackage{subdepth}
\usepackage{tensor}
\usepackage{mathtools}
\usepackage{textcomp}

\newcommand{\pState}[2]{\ket{\phi^{(#1)}}_{#2}}
\newcommand{\psiState}[2]{\ket{\psi^{(#1)}}_{#2}}
\newcommand{\en}[2]{\epsilon^{(#1)}_{#2}}

\usepackage{amsthm}
\newtheorem*{theorem}{Theorem}
\newtheorem{statement}{Statement}

\usepackage[capitalise]{cleveref}
\crefname{Remark}{remark\!}{Remark\!}
\crefname{equation}{Eq.\!}{Eqs.\!}
\crefname{chapter}{Chap.\!}{Chaps.\!}
\crefname{section}{Sec.\!}{Secs.\!}

\usepackage[export]{adjustbox}

\newcommand{\stepOne}{\textbf{Local resonator modes}}
\newcommand{\stepTwo}{\textbf{Edge state control}}
\newcommand{\stepThree}{\textbf{Contrast variation}}

\newcommand{\leftRightBraket}[3]{\tensor*[^{}_{#1}]{\braket{#2}}{^{}_{#3}}}
\newcommand{\leftBraket}[2]{\tensor*[^{}_{#1}]{\bra{#2}}{^{}_{}}}
\newcommand{\matlabFootnote}{To give precise results also for smallest amplitudes of around $10^{-20}$, the corresponding eigenstates have been obtained by diagonalizing $H$ with 64 decimal digits of precision by using the variable precision tools from the MATLAB\textsuperscript{\textregistered} Symbolic Math Toolbox\textsuperscript{TM}}

\begin{document}

\title{Local symmetry theory of resonator structures for the real-space control of edge states in binary aperiodic chains}

\author{M. Röntgen}
 \email[]{mroentge@physnet.uni-hamburg.de}
\affiliation{%
	Zentrum für optische Quantentechnologien, Universität Hamburg, Luruper Chaussee 149, 22761 Hamburg, Germany
}%

\author{C. V. Morfonios}%
\affiliation{%
	Zentrum für optische Quantentechnologien, Universität Hamburg, Luruper Chaussee 149, 22761 Hamburg, Germany
}%
\author{R. Wang}%
\affiliation{%
	Department of Electrical and Computer Engineering and Photonics Center, Boston University, Boston, MA 02215, USA	
}%
\author{L. Dal Negro}%
\affiliation{%
	Department of Electrical and Computer Engineering and Photonics Center, Boston University, Boston, MA 02215, USA	
}%
\affiliation{%
	Department of Physics, Boston University, Boston, MA 02215, USA
}%
\affiliation{%
Department of Material Science and Engineering, Boston University, Boston, MA 02215, USA
}%

\author{P. Schmelcher}
\affiliation{%
	Zentrum für optische Quantentechnologien, Universität Hamburg, Luruper Chaussee 149, 22761 Hamburg, Germany
}%
\affiliation{%
	The Hamburg Centre for Ultrafast Imaging, Universität Hamburg, Luruper Chaussee 149, 22761 Hamburg, Germany
}%

\begin{abstract}
	We propose a real-space approach explaining and controlling the occurrence of edge-localized gap states between the spectral quasibands of binary tight binding chains with deterministic aperiodic long-range order.
	The framework is applied to the Fibonacci, Thue-Morse and Rudin-Shapiro chains, representing different structural classes.
	Our approach is based on an analysis of the eigenstates at weak inter-site coupling, where they are shown to generically localize on locally reflection-symmetric substructures which we call local resonators.
	A perturbation theoretical treatment demonstrates the local symmetries of the eigenstates.
	Depending on the degree of spatial complexity of the chain, the proposed local resonator picture can be used to predict the occurrence of gap-edge states even for stronger couplings.
	Moreover, we connect the localization behavior of a given eigenstate to its energy, thus providing a quantitative connection between the real-space structure of the chain and its eigenvalue spectrum.
	This allows for a deeper understanding, based on local symmetries, of how the energy spectra of binary chains are formed.
	The insights gained allow for a systematic analysis of aperiodic binary chains and offers a pathway to control structurally induced edge states.
\end{abstract}

\maketitle

\section{Introduction} \label{sec:Introduction}

Aperiodic systems with deterministic long-range order have long been a subject of intense study, in the endeavor to systematically bridge the gap between crystalline periodicity and complete disorder \cite{Maciaroleaperiodicorder2006}.
While providing a powerful concept in theoretically modeling the transition to disorder, aperiodic order has become an established property of matter as well. 
A cornerstone of this was the actual observation of ``quasicrystals''---non-periodic but space-filling structures surpassing the crystallographic restriction theorem---by Shechtman \cite{ShechtmanMetallicPhaseLongRange1984}.
In nature quasiperiodicity occurs e.g. in macroscopic constellations such as phyllotaxis \cite{PennybackerPhyllotaxisPushedPatternForming2013,Maciaroleaperiodicorder2006}.
Aperiodically ordered systems even play an important role in material science and technology \cite{Maciaroleaperiodicorder2006,MaciaExploitingaperiodicdesigns2012}.
Owing to their long-range order, they can display interesting physical properties such as a low electrical and thermal conductance \cite{DuboisPropertiesapplicationsquasicrystals2012,Maciaroleaperiodicorder2006}, low friction\cite{DuboisPropertiesapplicationsquasicrystals2012,MancinelliTribologicalpropertiesB2type2003} and high hardness \cite{MancinelliTribologicalpropertiesB2type2003}. 
Specific quasicrystalline systems have been shown to enhance solar cells \cite{BauerLightharvestingenhancement2013}, serve as a catalyst\cite{YoshimuraQuasicrystalapplicationcatalyst2002} and could allow for superconductivity \cite{KogaFirstObservationSuperconductivity2015,KamiyaDiscoverysuperconductivityquasicrystal2018}.

A general characteristic of aperiodic lattices is the clustering of Hamiltonian eigenvalues into so-called ``quasibands'' resembling Bloch bands of periodic systems \cite{DalNegroStructuralSpectralProperties2016}.
The corresponding eigenstates generally neither extend homogeneously across the system like Bloch states in regular crystals, nor do they decay exponentially like in disordered systems, and are therefore dubbed ``critical'' \cite{OstlundRenormalizationgroupanalysisdiscrete1984,MaciaNatureElectronicWave2014,Kohmoto1987PRB351020CriticalWaveFunctionsCantorseta,FujiwaraMultifractalwavefunctions1989,Maciaroleaperiodicorder2006}.
In specific cases, quasibands have been shown to originate from the localization of different eigenstates on similar repeated substructures in the system yielding similar eigenenergies \cite{dePruneleFibonacciKochPenrose2001,dePrunelePenrosestructuresGap2002,BandresTopologicalPhotonicQuasicrystals2016,VignoloEnergylandscapePenrose2016,Macia2017PSSb2541700078ClusteringResonanceEffectsElectronic}.
The formation of quasibands typically becomes less distinct with increasing spatial complexity, which in turn can be classified by the structure's spatial Fourier transform---accordingly altering from point-like to singular continuous to absolutely continuous \cite{Maciaroleaperiodicorder2006,DalNegro2012LPR6178DeterministicAperiodicNanostructuresPhotonics,DalNegro2013OpticsAperiodicStructuresFundamentals,Macia2017APB5291700079SpectralClassificationOneDimensionalBinary}.
The Fourier spectrum can further be connected to the system's integrated density of states by the ``gap labeling theorem'' \cite{Johnsonrotationnumberalmost1982,Delyonrotationnumberfinite1983,LuckCantorspectrascaling1989,BellissardGaplabellingtheorems1992,BaakeTracemapsinvariants1993}, which assigns characteristic integers to the gaps between quasibands.

As ordered lattice systems are truncated in space into finite setups, they may support the occurrence of eigenstates localized along their edges, energetically lying within spectral gaps.
In periodic systems, such \emph{edge states} (or `surface states' \cite{DavisonBasicTheorySurface1996}) may or may not appear depending on how the underlying translation symmetry is broken by the lattice truncation, that is, where in the unit cell the system is cut off \cite{ZakSymmetrycriterionsurface1985}.
In various types of periodic setups, edge states can also be given a topological origin in terms of nontrivially valued invariants (winding numbers) assigned to the neighboring Bloch bands \cite{BernevigTopologicalInsulatorsTopological2013}.
This has boosted an intensive research activity in the field of topological insulators \cite{KaneTopologicalOrderQuantum2005,HasanColloquiumTopologicalinsulators2010,QiTopologicalinsulatorssuperconductors2011} and the quest for interesting novel materials and applications \cite{TianPropertyPreparationApplication2017}, including e.\,g. robust lasing via topological edge-states in periodic photonic lattices \cite{St-JeanLasingtopologicaledge2017}.

Edge states may also be present between quasibands in aperiodic systems, as has been shown for binary 1D systems \cite{ZijlstraExistencelocalizationsurface1999,ElHassouaniSurfaceelectromagneticwaves2006,LeiPhotonicbandgap2007,Pang2010JOSAB272009PhotonicLocalizationInterfaceModes,MartinezSurfacesolitonsquasiperiodic2012} and recently demonstrated for 2D photonic quasiperiodic tilings \cite{BandresTopologicalPhotonicQuasicrystals2016}.
Notably, also here a topological character can be assigned to the edge states in correspondence to the system's bulk properties.
Indeed, a position-space based topological invariant, the so-called Bott index \cite{LoringDisorderedtopologicalinsulators2010a}, can be applied to aperiodically structured \cite{BandresTopologicalPhotonicQuasicrystals2016} or even amorphous systems \cite{AgarwalaTopologicalInsulatorsAmorphous2017}.
Moreover, for 1D quasiperiodic systems the winding of edge state eigenvalues across spectral gaps coincide with the gap labels mentioned above \cite{Johnsonrotationnumberalmost1982,NegiCriticalstatesfractal2001,KellendonkRotationnumbersboundary2005,KrausTopologicalEquivalenceFibonacci2012}, which have recently also been measured in scattering \cite{BabouxMeasuringtopologicalinvariants2017} and diffraction \cite{DareauRevealingTopologyQuasicrystals2017} experiments.
Remarkably, edge modes occur also as scattering resonances in open systems with different types of deterministic aperiodic order incorporating long-range couplings between lattice constituents, as demonstrated very recently in terms of the eigenmodes of full vectorial Green matrices \cite{WangEdgemodesscattering2018}.

The ubiquitous presence of edge states in aperiodic systems indicates that it derives primarily from the underlying geometrical structure and not from model-specific assumptions.
Departing from periodicity, however, there is no translation symmetry whose breaking (at the boundary) would provide a mechanism for edge state formation.
On the other hand, aperiodic systems are imbued with \emph{local} symmetries, that is, different spatially symmetric substructures are simultaneously present in the composite system which possesses many different domains of local symmetries.
Indeed, local ``patterns'' are known to occur repeatedly in deterministic aperiodic systems, as expressed by Conway's theorem \cite{GardnerMathematicalGames1977}.
In the specific case of 1D binary lattices, \emph{local reflection symmetry} is abundantly present and follows, at each scale, a spatial distribution closely linked to the underlying aperiodic potential sequence \cite{Morfonios2014ND7871LocalSymmetryDynamicsOnedimensional}.
The encoding of such local symmetries into generic wave excitations have recently been described within a theoretical framework of symmetry-adapted nonlocal currents\cite{Morfonios2017AP385623NonlocalDiscreteContinuityInvariant}, which obey generalized continuity equations \cite{Morfonios2017AP385623NonlocalDiscreteContinuityInvariant,Rontgen2017AP380135NonlocalCurrentsStructureEigenstates,Spourdalakis2016PRA9452122GeneralizedContinuityEquationsTwofield,WulfExposinglocalsymmetries2016} and whose stationary form allows for a generalization of the parity and Bloch theorems to locally restricted symmetries \cite{Kalozoumis2014PRL11350403InvariantsBrokenDiscreteSymmetries} as well as a classification of perfect transmission \cite{Kalozoumis2013PRA8833857LocalSymmetriesPerfectTransmission}.
In the context of finite, aperiodically ordered setups, an appealing question is whether a real-space picture for the formation---and thereby control---of edge states can be brought into connection with local symmetries.

In the present work we propose an intuitive real-space picture of the formation of quasibands and edge states in binary aperiodic tight-binding chains.
The approach is based on the analysis of eigenstate profiles in the limit of weak inter-site coupling. In this regime, eigenstates generically fragment, i.e., have non-negligible amplitudes only on a small number of sites, as we show by means of a perturbation theoretical treatment.
The amplitudes on these fragments are in almost all cases locally symmetric and can be identified as \emph{local resonator modes} (LRM), i.e., eigenmodes of local resonators embedded into the full chain. Here, a resonator denotes a substructure that can confine, at certain energies, the wavefunction within its interior.
The LRMs can be used to classify states, and those belonging to quasibands are composed of repeated LRMs hosted by resonators within the bulk, while edge states are composed of unique LRMs occurring on the edge.
We further investigate the reasons for the formation of quasibands by linking the energy $\epsilon$ of a state to that of its constituting LRMs, where the energy of an LRM is defined as its energy in the corresponding \emph{isolated} resonator.
From this finding, we see that the multiple occurrence of identical resonator structures automatically leads to the formation of quasibands by their capability of hosting identical (and thus degenerate) LRMs.
We further use this energetical insight to move a given edge state into a quasiband by performing tailored changes to the corresponding resonators on the edge.
The inference of those properties to moderate inter-site coupling depends on the type of aperiodic order used in the model.
We here apply the approach to the prominent representatives of three main classes of structural complexity:
Fibonacci, Thue-Morse, and Rudin-Shapiro chains, featuring point-like, singular continuous, and purely singular spatial Fourier spectra, respectively.

The paper is organized as follows. 
In \cref{sec:setup} we introduce our setup and show examples of quasibands and edge states in Fibonacci chains. 
We then develop our approach to edge states based on locally symmetric resonators and apply it to Fibonacci chains in \cref{sec:edgeModes}, to Thue-Morse chains in \cref{sec:thue} and to Rudin-Shapiro chains in \cref{sec:rudin}. 
In \cref{sec:commentsOnGenerality} we comment on the generality of our framework and on the connection to related work.
We conclude our paper and give an outlook in \cref{sec:conclusions}. 
A perturbative treatment demonstrating the localization onto reflection-symmetric resonators is provided in the appendix, together with further technical details including proofs of major statements, complementary explanations, and further comments.

\section{Prototype quasiperiodic order: The tight-binding Fibonacci chain} \label{sec:setup}

We consider a finite one-dimensional chain of $N$ sites with real next-neighbor hoppings $h_{m,n}$ described by the Hamiltonian
\begin{equation}
 H = \sum_n v_n \ket{n}\bra{n} + \sum_{|m-n| = 1} h_{m,n} \ket{m}\bra{n}
\end{equation}
where $v_n$ is the onsite potential of site $n$. In the basis of single site excitations $\ket{n}$, the above Hamiltonian $H$ can be written as a tridiagonal matrix
\begin{equation} \label{eq:tridiagonalHamiltonian}
	H = \begin{pmatrix}
	v_{1}	& h_{1,2}	& 0 & \dots	 & 0      \\
	h_{1,2}	& v_{2} 	& h_{2,3}  & \ddots & \vdots	  \\
	0 	& h_{2,3} 	& \ddots & \ddots & 0\\
	\vdots 	& \ddots & \ddots	 & \ddots & h_{N-1,N} \\
	0 	& \dots & 0	 & h_{N-1,N} & v_{N},
	\end{pmatrix} .
\end{equation}
Such a tight-binding chain is used in a plethora of interesting model systems, examples including the Aubry-Andre \cite{AubrySergeAnalyticitybreakingAnderson1980} model relevant in the study of localization\cite{roatiAndersonLocalizationNoninteracting2008} and the Su-Schrieffer-Heeger model, a simple prototypical chain supporting a topological phase \cite{asbothShortCourseTopological2016}. 
It also effectively describes one-dimensional arrays of evanescently coupled waveguides \cite{Szameit2012DiscreteOpticsFemtosecondLaser,EfremidisWavepropagationwaveguide2010}.
We here fix the hoppings to a uniform value $h$ and restrict the onsite elements to be ``binary'', that is, the sites are of two possible types $A$ and $B$, and the $v_n$ take on corresponding values $v_A$ and $v_B$, with the \emph{contrast} defined as
\begin{equation} \label{eq:contrast}
 c = \left|\frac{v_A - v_B}{h}\right|.
\end{equation}
Without loss of generality we will set $v_A \equiv 0$ and $v_B \equiv v$ throughout, having a single control parameter $c = |v/h|$ for a given chain.

In the following, we will investigate the spatial profiles of the eigenvectors $\ket{\phi^\nu} = \sum_n \phi^\nu_n \ket{n}$ of $H$ in relation to their eigenvalues $\epsilon_\nu$, given by
\begin{equation}
 H \ket{\phi^\nu} = \epsilon_\nu \ket{\phi^\nu},
\end{equation}
for chains with aperiodic order.
Note that $H$ represents a generic finite tight-binding chain;
the choice $h<0$ corresponds, e.g., to the kinetic energy of electrons on a lattice with onsite potential $v_n$, while $h>0$ (made here) can be used to model the coupling of photonic waveguides \cite{Szameit2012DiscreteOpticsFemtosecondLaser,EfremidisWavepropagationwaveguide2010} with propagation constants $v_n$.
Our analysis remains qualitatively unaffected by this choice.

\begin{figure}[!] 
\centering
\includegraphics[max size={1\columnwidth}{1\textheight}]{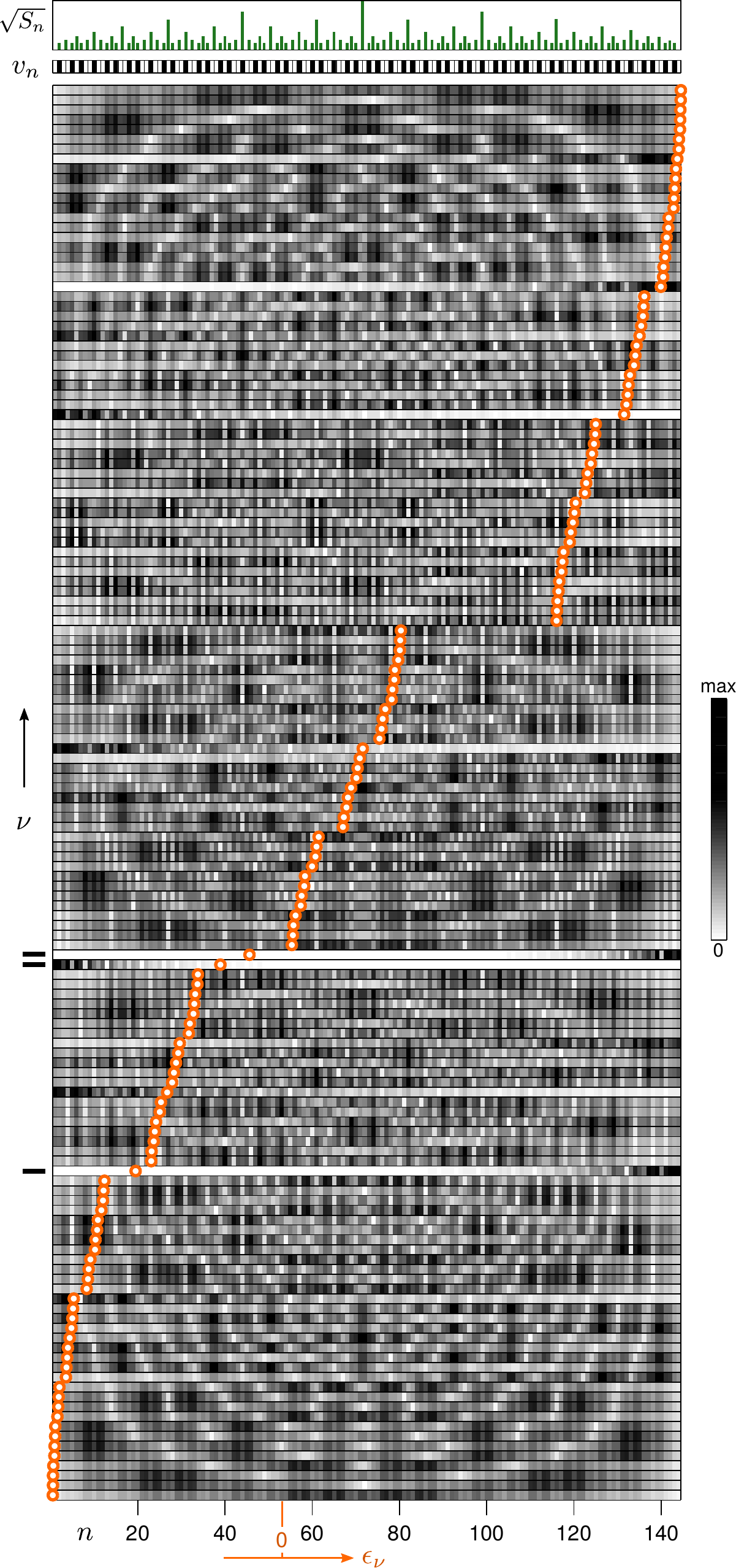}
\caption{
\textbf{Bottom}: 
Eigenstate map of a $N=144$-site Fibonacci chain at contrast $c = |v/h| = 1.5$ (hopping $h=0.1$): 
Each horizontal stripe shows $\sqrt{|\phi^\nu_n|}$ at sites $n$ for an eigenstate $\phi^\nu$ ($\nu = 1,2,\dots,N$).
The greyscale map is chosen such that it is possible to simultaneously observe the spatial features of edge as well as those of bulk states.
Superimposed are the eigenvalues $\epsilon_\nu$ (orange circles) in arbitrary units, with indicated origin $\epsilon = 0$. 
Edge modes are distinguishable as partially white stripes, with the most pronounced ones indicated by black horizontal bars on the left.
\textbf{Middle}: Potential $v_n$ represented by a stripe with white (black) boxes for $v_n = v_A = 0$ ($v_n = v_B = v$).
\textbf{Top}: Distribution of local reflection symmerty domains, represented by maximal domain size $S_n$ centered at position $n$, as explained in the text.}
\label{fig:fibo144StateMap}
\end{figure}

We start by presenting the eigenstates and spectral properties of a finite binary chain following the Fibonacci sequence \cite{DalNegro2013OpticsAperiodicStructuresFundamentals}, a prototypical case of quasiperiodic order.
This will serve as an initial point motivating the development of a local resonator approach at high contrast in the next section.
Starting with $A$, the sequence is constructed by repeatedly applying the inflation rule $A \to AB, B \to A$, resulting in $F = ABAABABAAB...$. 
This sequence is then mapped onto the onsite elements $v_{n}$ of the tight-binding chain.
The spectrum and eigenvectors of this chain are shown in \cref{fig:fibo144StateMap} for a moderate contrast of $c = 1.5$ and $N=144$ sites.
Despite the lack of periodicity, the eigenvalues cluster into so-called quasibands, owing to the long-range order present in the Fibonacci chain \cite{MaciaBarber2008AperiodicStructuresCondensedMatter}, and the spectrum attains a self-similar structure of quasibands and gaps in the $N \to \infty$ limit. For presentation reasons, we have here chosen $N$ large enough to anticipate this spectral feature, though small enough to visually discern the spatial characteristics of the eigenmodes.

The quasibands are occupied by bulk eigenmodes that extend along the interior of the chain. Those are known as ``critical states'', with a spatial profile lying between the exponential decay of modes in a randomly disordered chain and uniformly extending Bloch eigenmodes in periodic chains\cite{Kohmoto1987PRB351020CriticalWaveFunctionsCantorseta,FujiwaraMultifractalwavefunctions1989,MaciaNatureElectronicWave2014}.
Such modes have recently been shown to consist of locally resonating patterns (i.e., characteristic sequences of amplitudes) which occur on repeating segments of a quasiperiodic structures and are characteristic for a given quasiband \cite{VignoloEnergylandscapePenrose2016,dePrunelePenrosestructuresGap2002,MaceFractaldimensionswave2016}.
This is particularly visible for the bulk modes of the uppermost quasiband in \cref{fig:fibo144StateMap}.
A close inspection reveals that the bulk mode profiles tend to localize into locally reflection-symmetric peaks (see black subregions of high amplitude for a given mode). 
Those in turn follow the distribution of local symmetry axes (or centers of ``palindromes'' \cite{DroubayPalindromesFibonacciword1995}) which are hierarchically present in the Fibonacci chain \cite{Morfonios2014ND7871LocalSymmetryDynamicsOnedimensional}, as seen by comparison with the bar plot on the top.
Each bar shows the maximal size $S_n$ of a continuous domain of reflection symmetry centered at position $n$, where $n$ can refer here to sites ($n = 1,2,...$) or to links between sites ($n = 1.5,2.5,...$).
For instance, as the first few characters of $F$ are 
\begin{equation*}
	\mathrlap{\overbrace{\phantom{ABAABA}}^{\text{6 sites}}}
	ABA
	\mathrlap{\underbrace{\phantom{ABA}}_{\text{3 sites}}}
	ABABAAB,
\end{equation*}
we have $S_{3.5} = 6$ and $S_5 = 3$.

Within the gaps between quasibands there may appear spectrally isolated modes, reminiscent of gap modes localized on defects within a periodic lattice \cite{GrundmannPhysicsSemiconductorsIntroduction2016,PovinelliEmulationtwodimensionalphotonic2001}. For the example given in \cref{fig:fibo144StateMap}, i.e., an unperturbed but finite Fibonacci chain, the gap modes are known \cite{niuSpectralSplittingWavefunction1990,ZijlstraExistencelocalizationsurface1999,NoriAcousticelectronicproperties1986} to be localized at the edges, decaying exponentially into the bulk.

The control of edge states by local changes in the underlying potential sequence is a central aspect of this work.
Our approach is that, due to their localization, the occurrence and spectral position of edge states can be influenced by local modifications on the corresponding edge of the aperiodic lattice.
We demonstrate the feasibility of this approach in \cref{fig:fiboPhasonFlips} by using the following representation \cite{BabouxMeasuringtopologicalinvariants2017} of the Fibonacci potential sequence:
\begin{align}
 v_n = \frac{v_A + v_B}{2} + \frac{v_A - v_B}{2} \text{sign}\chi_n = \frac{v}{2}(1 - \text{sign}\chi_n), \label{eq:fiboPot} \\
 \chi_n(\varphi) = \cos( 2\pi\tau n + \varphi + \pi\tau ) - \cos( \pi\tau )
\end{align}
where $\tau = 2/(1 + \sqrt{5})$ is the inverse golden mean and the integer site index $n$ runs from $1$ to $N$.
By continuously varying the so-called ``phason'' $\varphi$, localized flips $AB \leftrightarrow BA$ are induced at discrete values of $\varphi$, forming a two-dimensional pattern in the $(n,\varphi)$ plane, see \cref{fig:fiboPhasonFlips}.
The finite chain of length $N$ constitutes a different segment (or ``factor'') of the infinite Fibonacci sequence after each flip \cite{BabouxMeasuringtopologicalinvariants2017}. This allows to investigate different Fibonacci-like configurations while maintaining a constant length $N$. In \cref{fig:fiboPhasonFlips}, we visualize the effect of these flips on the energy spectrum, shown in orange. As one can see, the gap states in the purple rectangle, which are localized on the right edge (not shown here), are influenced only by flips acting on this edge, marked by green circles. From bottom to top, the green flips (i) create the edge state (ii) and (iii) modify its energy and (iv) finally annihilate it. Note that in general for processes of type (ii) and (iii), the energetical change accompanying the change of the edge is stronger for a flip near to the edge than for a flip more distant to the edge.

\begin{figure}[t]
	\centering
	\includegraphics[max size={.99\columnwidth}{.7\textheight}]{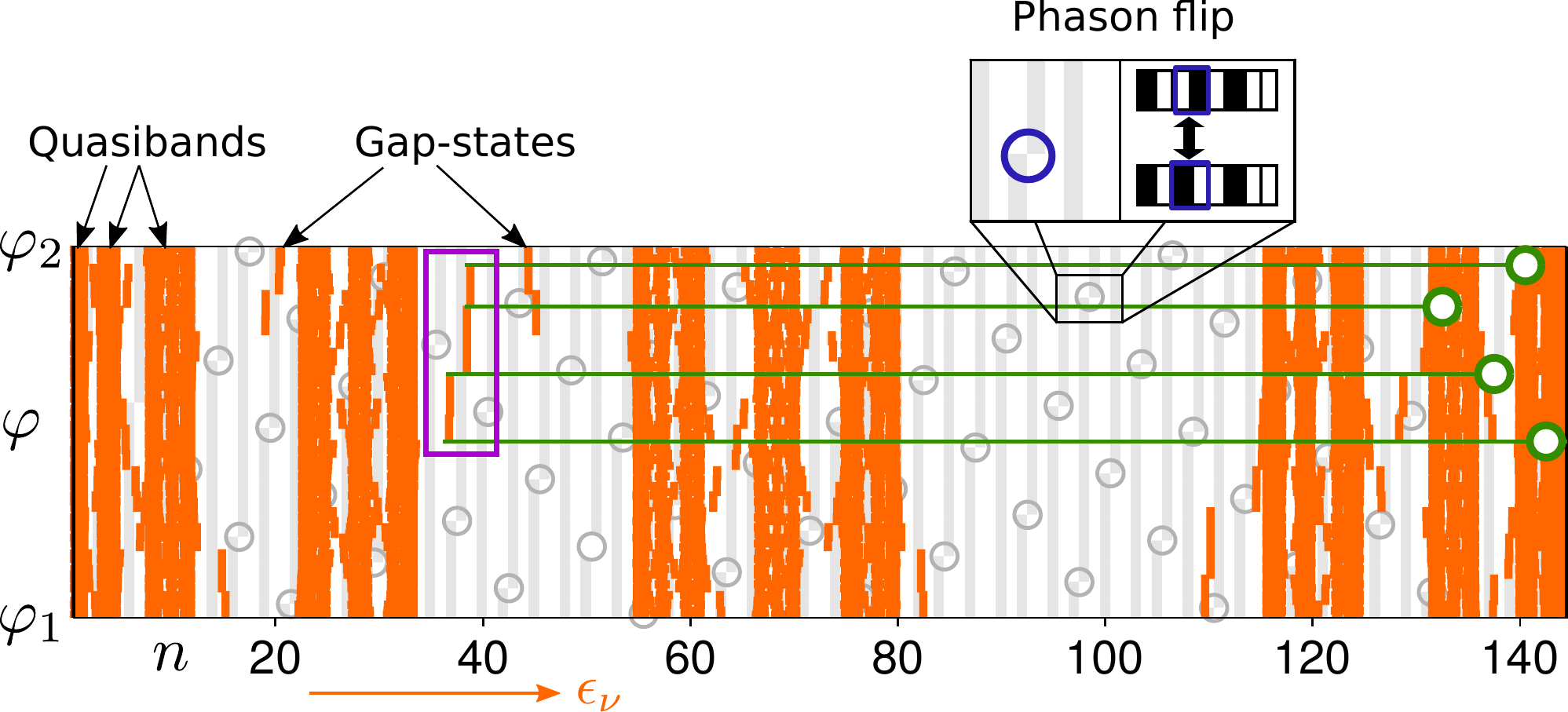}
	\caption{
 Spectrum in arbitrary units (orange) of a $N=144$-site Fibonacci chain for varying phason $\varphi$ in \cref{eq:fiboPot} between the values (chosen for presentation reasons) $\varphi_{1} = 2.4097$ and $\varphi_{2} = 5.5513$ for contrast $c = 1.5$, superimposed on the variation of the onsite potential $v_n$ ($v_A$:white, $v_B$:light gray).
 Dark gray circles indicate local flips $AB \leftrightarrow BA$ in the chain. The inset shows a magnified view on one representative flip. All together, there are $71$ such flips between $\varphi_{1}$ and $\varphi_{2}$.
The flips indicated by green circles create/annihilate (when close to the edge) or energetically shift (when further from the edge) the gap state in the purple box.
	}
	\label{fig:fiboPhasonFlips}
\end{figure}

The occurrence of such edge-localized gap states in a finite 1D quasiperiodic potential was recently very elegantly described within a scattering setting \cite{BabouxMeasuringtopologicalinvariants2017,DareauRevealingTopologyQuasicrystals2017} in a continuous system as a consequence of a resonance condition when varying the phason $\varphi$.
At the same time, the connection of the winding of $\varphi$ to invariant integers labeling the spectral gaps of the quasiperiodic structure through the so-called ``gap labeling'' theorem \cite{FuPerfectselfsimilarityenergy1997}, renders the nature of the 1D edge states topological \cite{KrausTopologicalEquivalenceFibonacci2012}.
On the other hand, the flip-induced edge state creation/annihilation demonstrated in \cref{fig:fiboPhasonFlips} suggests that their origin could also be explained by viewing chain edges as a generalized type of ``defects'' to the quasiperiodic long-range order.
In the following, we will develop this idea in terms of the prototype Fibonacci chain.
Our aim is to provide a simple and unifying real-space picture for the appearance of edge states in the energy gaps of non-periodic structures.
Contrary to topological methods, as employed for one-dimensional systems in general e.g., in Refs. \onlinecite{LevyTopologicalpropertiesFibonacci2015,KrausTopologicalEquivalenceFibonacci2012,Verbin2015PRB9164201TopologicalPumpingPhotonicFibonacci,Baake2012JMP5332701SpectralTopologicalPropertiesFamily,Parto2018PRL120113901EdgeModeLasing1DTopological,Blanco-Redondo2016PRL116163901TopologicalOpticalWaveguidingSilicon,Johnsonrotationnumberalmost1982,NegiCriticalstatesfractal2001,KellendonkRotationnumbersboundary2005,BabouxMeasuringtopologicalinvariants2017,DareauRevealingTopologyQuasicrystals2017}, our approach does not rely on topology, but aims at connecting the real-space structure of deterministic aperiodic binary chains and their local symmetries to their quasibands and edge states.

\section{Edge modes from truncated local resonators} \label{sec:edgeModes}

The analysis of eigenstates at high contrast $c$ [see \cref{eq:contrast}] is at the heart of our approach, revealing structural information that would remain hidden at lower contrast.
Once this information is retrieved, we leverage it to develop a generic framework for the understanding and manipulation of quasibands and edge states in binary tight-binding chains.

In the following, we will focus on a Fibonacci chain, choosing a relatively small size for easier treatment and visualization. The slight modifications needed for the treatment of longer chains are commented on in \cref{appendix:longerChainsComments}.
We split the presentation into three subsections, covering the concept of fragmentation (\cref{sec:degPertTheory}), local resonator modes (\cref{sec:locResAndLocSym}), the structural control of edge states (\cref{sec:controlOfEdgeStates}), and the behavior at low contrast (\cref{sec:lowContrastBehavior}).

\subsection{Eigenstate fragmentation from degenerate perturbation theory} \label{sec:degPertTheory}
\begin{figure}[ht!]
\centering
\includegraphics[max size={.99\columnwidth}{0.7\textheight}]{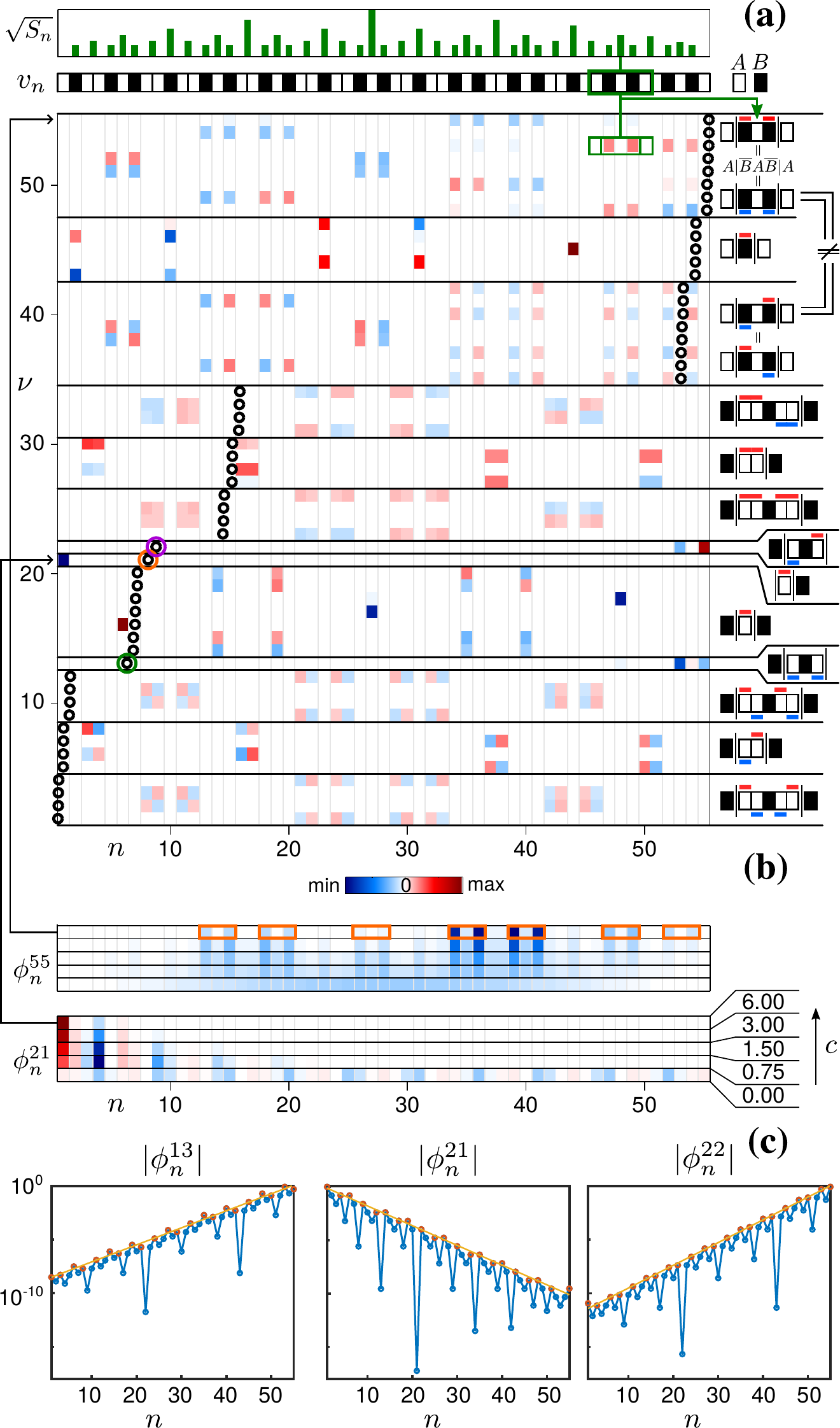}
\caption{
\textbf{(a)} Eigenstate map, potential, and local symmerty distribution (bottom to top) like in \cref{fig:fibo144StateMap} but for an $N=55$-site Fibonacci chain of contrast $c = 6$, and with density $|\phi_{n}^{\nu}|^2$ color-coded by the signs of $\phi_{n}^{\nu}$ shown in the eigenstate map. 
Horizontal lines separate the eigenstates into groups according to quasibands and gap states, with corresponding characteristic local resonator modes (LRMs) visualized on the right.
The green box indicates the correspondence of the state $\phi^{53}$ to the LRM $A|\overline{B}A\overline{B}|A$ (see text). 
The three edge modes of the setup are marked by colored circles.
\textbf{(b)} Amplitudes of states $\phi^{55}$ (extended in the bulk) and $\phi^{21}$ (localized at the left edge) for different contrast values. \textbf{(c)} Absolute values of amplitudes of the three edge states $\phi^{13}$, $\phi^{21}$, $\phi^{22}$ with corresponding localization lengths $0.37$, $0.43$, $0.49$, obtained by fitting a line (orange) to local maxima (orange dots) on a logarithmic scale.}
\label{fig:fibo55StateMap}
\end{figure}

Our starting point is an analysis of eigenstates at high contrast.
Those are shown in \cref{fig:fibo55StateMap}\,(a) for a $9$th generation Fibonacci chain ($N=55$ sites) with relatively high contrast $c=6$. 
We see that each eigenstate is pinned to a small number of sites where it has non-negligible amplitude, practically vanishing on the remaining sites.
This is quite different from the states at low contrast (like in \cref{fig:fibo144StateMap}) which are smeared out along the whole chain.
An impression of how the transition between those two regimes takes place is given in \cref{fig:fibo55StateMap}\,(b), showing the amplitudes of a bulk $(\phi^{55})$ and edge $(\phi^{21})$ state for varying contrast.
When increasing the contrast, the spatial profile of the bulk state becomes gradually \emph{fragmented}:
The amplitudes on $A$-sites become suppressed, and a characteristic remnant of the initial distribution appears on a subset of $B$-sites.
Fragmentation with increasing contrast $c$ also occurs for the edge state, with the difference that here the amplitudes on $B$ sites become suppressed, and that there is only a \emph{single} fragment remaining; in the present case the $A$-site on the left edge.

The fragmentation at high contrast can be understood by means of a quantitative perturbation-theoretical treatment provided in \Cref{app:perturbationTheory}, applying to generic binary tight-binding chains. In the following, we outline the main steps of this analysis.
In order to apply perturbation theory, the Hamiltonian is written as $H_{0} + h \cdot H_{I}$, where $H_{0}$ solely contains the diagonal part of $H$, i.e., isolated sites, while $H_{I}$ has $1$'s only on the first sub- and superdiagonal.
For convenience, we then rescale $H' = H/v = H'_{0} + 1/c \cdot H_{I}$, changing only the energies $\epsilon^{\nu} \rightarrow \epsilon^{\nu}/v$, but leaving all eigenstates unaffected. For large contrast $c$, $H_{I}$ then acts as perturbation to $H_{0}$, and we can expand an eigenstate $\ket{\phi}^{(i)}\; (1 \le i \le N)$ of $H \in \mathbb{R}^{N\times N}$ as well as its energy $\en{i}{}$ as
\begin{align} \label{eq:perturbationSeriesStatesMain}
\pState{i}{} & = \pState{i}{0}  + \lambda \pState{i}{1}  + \lambda^2 \pState{i}{2}  + \ldots \\
\en{i}{} &= \en{i}{0} + \lambda \en{i}{1} + \lambda^2 \en{i}{2} + \ldots \label{eq:perturbationSeriesEnergiesMain},
\end{align}
which, inserted into the Schrödinger equation, yields the perturbation series.
Due to the binary nature of $H_{0}$, the only two eigenvalues of $H_{0}$, $0$ and $1$, are highly degenerate. In particular, before any higher-order state correction can be computed, the so-called ``correct zeroth-order states''\cite{Hirschfelder1974JCP601118DegenerateRSPerturbationTheory,Silverstone1981PRA231645PracticalRecursiveSolutionDegenerate}
\begin{equation}
\pState{i}{0} = \lim\limits_{\lambda \rightarrow 0}{\pState{i}{}}
\end{equation}
must be found. Although these are superpositions of the known eigenstates of $H_{0}$, the corresponding expansion coefficients are in general unknown at the beginning of the treatment\cite{Hirschfelder1974JCP601118DegenerateRSPerturbationTheory,Silverstone1981PRA231645PracticalRecursiveSolutionDegenerate}.
In degenerate perturbation theory\cite{Hirschfelder1974JCP601118DegenerateRSPerturbationTheory,Silverstone1981PRA231645PracticalRecursiveSolutionDegenerate}, the correct zeroth-order states can be found by diagonalizing a series of recursively\cite{Silverstone1981PRA231645PracticalRecursiveSolutionDegenerate} defined matrices $\mathcal{H}_{1},\mathcal{H}_{2},\ldots{}$\, .
More precisely, the matrices $\mathcal{H}_{n}$ are constructed from the perturbation series up to $n$-th order by demanding that the correct zeroth-order states fulfill certain consistency requirements. One then has to solve
\begin{equation} \label{eq:consistencyRequirements}
	\leftRightBraket{0}{\phi^{(i)}|\mathcal{H}_{n}|\phi^{(j)}}{0} = \delta_{i,j} \en{i}{n}, \; \forall \; i,j \in g_{n}
\end{equation}
up to the order $n$ in which all degeneracies are lifted, where $\pState{g_{n}}{}$ is the set of states which are degenerate up to $n$-th order.
Now, contrary to simple examples where the degeneracy is resolved in first order (where $\mathcal{H}_{1}$ is simply given by  $H_{I}$), the degeneracies of binary chains are usually completely resolved only in higher orders.
As a result, the procedure of obtaining the correct zeroth-order states is rather complex\cite{Hirschfelder1974JCP601118DegenerateRSPerturbationTheory,Silverstone1981PRA231645PracticalRecursiveSolutionDegenerate}.

In \cref{app:perturbationTheory}, we explicitly follow this procedure of finding the correct zeroth-order states up to third order and investigate the first-order state corrections as well.
This procedure provides a high degree of understanding of how binary chains, their local symmetries, the fragmentation of eigenstates as well as their symmetries are connected.
In particular, it is shown that each $\pState{i}{0}\; (1\le i \le N)$, with $N$ being the length of the chain, has non-vanishing amplitudes either only on either $A$-sites or only on $B$-sites (see Statement 2 of \cref{app:perturbationTheory}).
Thus, we can assign each $\pState{i}{0}$ a type $T \in \{A,B\}$, depending on the sites on which it has non-vanishing amplitudes.
The spatial distribution of the non-vanishing amplitudes can be further specified by introducing the concept of \emph{maximally extended blocks of potentials of the same kind} (MEBPS). An example for such MEBPS are 
\begin{equation*}
	\overbracket{A}^1\underbracket{B}_{1}\overbracket{AA}^{2}\underbracket{B}_{1}\overbracket{A}^{1}\underbracket{B}_{1}\overbracket{AAA}^{3}
\end{equation*}
where MEBPS of type $A$ ($B$) are marked by over (under) brackets, with respective length denoted by integers.
An important conclusion of the analysis is that a given correct zeroth-order state $\pState{i}{0}$ of type $T$ can have non-vanishing amplitude on MEBPS of type $T$ and of individual length $l_{1},l_{2},\ldots{},l_{n}$ only if there exist integers $1 \le k_{j} \le l_{j},\; 1 \le j \le n$ such that (see Statement 2 of \cref{app:perturbationTheory})
\begin{equation*}
\frac{k_{1}}{l_{1} + 1} = \frac{k_{2}}{l_{2} + 1} = \ldots{} = \frac{k_{n}}{l_{n} + 1}.
\end{equation*}
As one can easily show, for $l_{j}\le 6$, this is possible only if all $l_{j}$ are identical or if all $l_{j}$ are odd. As a consequence, for the Fibonacci setup, where $l_{j} \le 2$, any $\pState{i}{0}$ can simultaneously have non-vanishing amplitudes only on MEBPS of length $1$ \emph{or} of length $2$, but not on both. As a result, any state $\pState{i}{0}$ has vanishing amplitudes on a large number of sites, ultimately leading to its fragmentation. A closer evaluation reveals that this fragmentation usually persists under inclusion of the first-order state corrections $\pState{i}{1}$: If $\pState{i}{0}$ has non-vanishing amplitudes only on $A$ ($B$) sites, then $\pState{i}{1}$ will have non-vanishing amplitudes only on a small number of $B$ ($A$) sites.
As, at high contrast, $\pState{i}{} \approx \pState{i}{0} + \pState{i}{1}$, our perturbation theoretical treatment thus explains the origin of the fragmentation of eigenstates in binary tight-binding chains in a rigorous quantitative way.
Compared to the renormalization group approach which has been used\cite{kaluginElectronSpectrumOnedimensional1986,piechonAnalyticalResultsScaling1995,liuBranchingRulesEnergy1991,niuSpectralSplittingWavefunction1990,niuRenormalizationGroupStudyOneDimensional1986} to explain the fractal nature of the spectrum of the Fibonacci chain and which needs to be tailored to the system of interest, we stress that our perturbation theoretical approach is much broader and can be used to treat all binary chains where fragmentation occurs.
We demonstrate this generality by further analyzing the spatial details of those fragmented states and connecting them to the local symmetries of the chain and their environment (neighboring sites) in \cref{app:perturbationTheory} (see Statements 3 and 4 as well as following text).

\subsection{Local resonator modes and local symmetry} \label{sec:locResAndLocSym}
Relying on the above perturbation theoretical results, we now promote an intuitive picture for the cause of fragmentation, where a chain is viewed as a collection of \emph{local resonators}. The eigenvalues of this chain are then approximately given by the union of the eigenvalues of the individual resonators. As a consequence, each eigenvector of the full chain with energy $\epsilon$ then has very small amplitude on resonators whose energy differs strongly from $\epsilon$.
A local resonator is here a discrete substructure which, at high contrast, confines the wavefunction within its interior for a certain eigenenergy.
The simplest case consists of a three-site structure $B|A|B$, where the vertical lines demarcate the resonator ``cavity'' (the inner part $A$) from its ``walls'' (the outer parts $B$).
The resonator character of this particular substructure is analyzed in more detail in \cref{appendix:discreteResonators}.
Two such resonators can be combined to form a \emph{double} resonator $B|ABA|B$, formed by overlapping one wall of each $B|A|B$. 
Note that, for a substructure to function as a local resonator, either (i) the resonator wall and its next-neighboring site in the cavity must be of different type or (ii) the resonator wall must coincide with one of the edges of the chain ($|X$ or $X|$, with $X = A,B$).

We now link the resonator concept to the eigenstate fragmentation seen in \cref{fig:fibo55StateMap} (a).
As an example, each fragment of $\phi^{55}$ (indicated by orange rectangles in \cref{fig:fibo55StateMap}\,(b)) is localized on the $B$'s of the local  resonator $A|BAB|A$.
We denote this fact as $A|\overline{B}A\overline{B}|A$, which represents an eigenmode of the isolated resonator $A|BAB|A$ and which we will call a \emph{local resonator mode} (LRM).
The overlines here indicate sites with equally signed and relatively much higher amplitude than non-overlined sites; see \cref{appendix:discreteResonators}.
At high contrast the state $\phi^{55}$ can thus be seen as a collection of identical, non-overlapping LRMs $A|\overline{B}A\overline{B}|A$ (one on each fragment) with negligible amplitudes on the parts in between.
In the same manner, each eigenstate shown in \cref{fig:fibo55StateMap}\,(a) is composed of identical LRMs.
In particular, we notice that \emph{all states in a given quasiband are characterized by the same resonator mode}, different from that of other quasibands.
This is shown on the right side of the figure, where LRMs are depicted schematically.
Here, overlines and underlines in an LRM such as $A|\overline{B}A\underline{B}|A$ denote amplitudes of opposite sign.
Contrary to the bulk states of quasibands, \emph{edge states feature unique resonator modes} which are not repeated elsewhere in the chain, with the underlying resonators located at (one of) the chain edges.
We thereby distinguish these two types of LRMs as bulk and edge LRMs ($b$-LRMs and $e$-LRMs, respectively).

The fact that each quasiband is characterized by a single resonator mode can be understood as follows.
If a given eigenstate $\ket{\phi}$ of energy $\epsilon$ is composed of $K$ non-overlapping LRMs such that $\ket{\phi}$ has very low amplitude on the next-neighboring sites of the corresponding resonators, then each of the energies $\epsilon_{k=1,2,\dots,K}$ of those LRMs (that is, their eigenenergies in the \emph{isolated} underlying local resonator) must fulfill $\epsilon_{k} \approx \epsilon$.
This statement is proven rigorously in \cref{appendix:submatrixTheorem}.
Now, applying the perturbative treatment of \cref{app:perturbationTheory} to the chain of \cref{fig:fibo55StateMap} shows that for any two LRMs to be energetically nearly degenerate they must be identical.
Thus, each quasiband---having quasidegenerate modes at high contrast---is characterized by a single LRM.

A similar reasoning explains why bulk states of quasibands are more spatially extended compared to edge states lying in spectral gaps.
Indeed, due to the quasiperiodicity of the Fibonacci chain, any local resonator (that is, a binary substructure) in the bulk occurs \emph{repeatedly} (though not periodically) along the chain---specifically, at spacings smaller than double its size.
This is a general result known as Conway's theorem \cite{GardnerMathematicalGames1977}.
Thus, a $b$-LRM hosted by a given local resonator will also be correspondingly repeated along the chain.
If the $b$-LRM has energy $\epsilon_k$, then a state with energy $\epsilon \approx \epsilon_k$ is allowed to simultaneously occupy all copies of this $b$-LRM, and is accordingly spatially extended.
Edge states, on the other hand, consisting of $e$-LRMs at high contrast, correspond to local resonators induced by the presence of an edge, which \emph{breaks the quasiperiodicity}.
Due to this truncation at the edge (e.\,g. of the type $\cdots | \cdots X|$ with $X = A,B$ at the right edge, compare \cref{fig:fibo55StateMap} (a)), the $e$-LRM generally does not match the energy of any $b$-LRM, and therefore cannot occupy multiple local resonators in the bulk:
The eigenstate is confined to the edge, lying energetically isolated in a gap.
This is visualized by the marked edge states in \cref{fig:fibo55StateMap} (a).

A remarkable observation in \cref{fig:fibo55StateMap}\,(a) is that each local resonator hosting a $b$-LRM is reflection symmetric, such that all isolated $b$-LRMs have definite parity; see schematic on the right.
This means that, at high contrast, the fragments (occupied local resonators) of quasiband eigenstates feature, to a very good approximation, local parity with respect to \emph{local reflection symmetries} of the chain.
The positions and sizes of all such local symmetries are shown in the top panel of \cref{fig:fibo55StateMap}\,(a).
An example is given by the state $\phi^{53}$ which is locally symmetric around, e.\,g., the position $n = 48$, and corresponds to the $b$-LRM $A|\overline{B}A\overline{B}|A$, as indicated by the green boxes.
This behavior is predicted by the degenerate-perturbative treatment of \cref{app:perturbationTheory}. There, we rigorously show that each $\pState{i}{0}$ is locally parity symmetric individually on each MEBPS (see Statement 1 of \cref{app:perturbationTheory}), which itself is by definition a locally symmetric structure. While an MEBPS usually comprises only a few sites, we explicitly give examples for cases where $\pState{i}{0}$ is locally symmetric also in larger domains.
One of this examples explains the local symmetry of LRMs such as $A|\overline{B}A\overline{B}|A$ (see Statement 3 of \cref{app:perturbationTheory}).
Overall, the perturbation theoretical treatment demonstrates the crucial role of local reflection symmetries in the eigenstate localization profiles of binary aperiodic chains.
A promising direction of research would thus be to treat this class of systems within the recently developed theoretical framework of local symmetries \cite{Kalozoumis2014PRL11350403InvariantsBrokenDiscreteSymmetries,Morfonios2017AP385623NonlocalDiscreteContinuityInvariant,Rontgen2017AP380135NonlocalCurrentsStructureEigenstates}.

\subsection{Structural control of edge states} \label{sec:controlOfEdgeStates}

\begin{figure}[t]
\centering
\includegraphics[max size={1\columnwidth}{.6\textheight}]{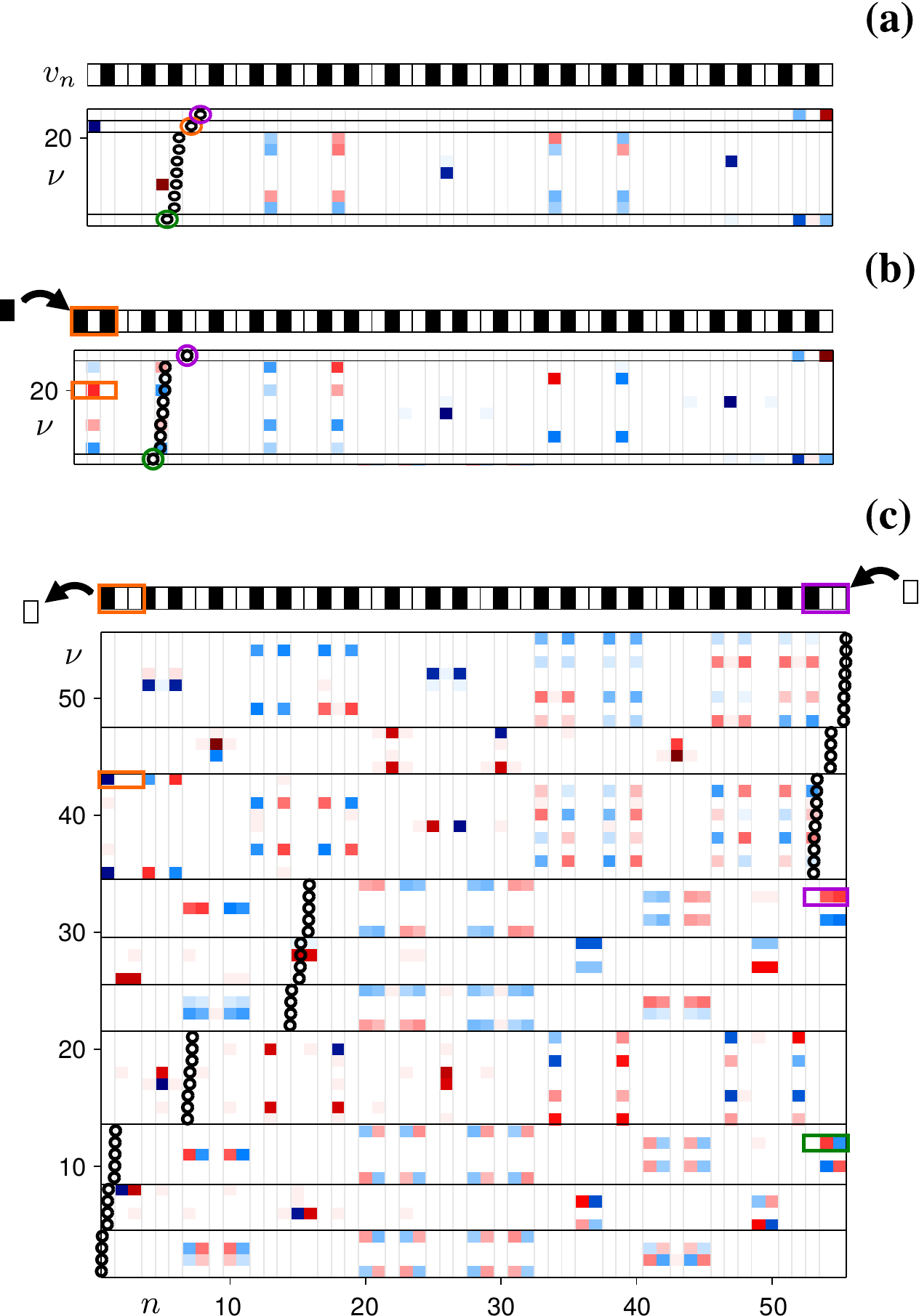}
\caption{
Selective annihilation of edge states of the Fibonacci chain in \cref{fig:fibo55StateMap}\,(a); see text for details.
\textbf{(a)} Original chain and excerpt of the state map (states $13$--$22$), with edge states marked by colored circles. 
\textbf{(b)} Annihilation of left edge state (orange) by attaching a $B$ site to the left of the chain. 
\textbf{(c)} Complete annihilation of edge states by removing (adding) an $A$ ($B$) site on the left (right) end of the chain. Color coding of each subfigure is as in \cref{fig:fibo55StateMap}.}
\label{fig:fibo55Modified}
\end{figure}

Having understood the real-space mechanism for the formation of edge-localized gap states in Fibonacci chains, we can now utilize this insight to systematically manipulate these states.
In particular, let us show how structural modifications at the edges of a Fibonacci chain can selectively ``annihilate'' a given edge mode.
Note that whether or not one considers a particular state localized (near or on) the edge to lie in an energetical gap is obviously a question of the scale under consideration.
This is due to the fact that any finite chain naturally has a discrete spectrum, for which, strictly speaking, no continuous energy-bands are defined.
In the remainder of this work, we will solely consider states as gap-edge ones provided that, at a contrast of $c = 6$, they lie in a clearly visible energetical gap.
This simplifies our treatment, and in \cref{appendix:longerChainsComments} we comment on the extensions of this simplification.

For definiteness, we consider the edge state $\phi^{21}$ (orange circle) of the chain in \cref{fig:fibo55Modified}\,(a) which simply focuses on states $\nu = 13$ to $22$ of \cref{fig:fibo55StateMap}\,(a).
This state corresponds to the the $e$-LRM $|\overline{A}|B$ (see right side of \cref{fig:fibo55StateMap}\,(a)) and is exponentially localized, as shown in \cref{fig:fibo55StateMap} (c). 
The underlying resonator $|A|B$ is a left-truncated version of the resonator $B|A|B$, which hosts the $b$-LRM $B|\overline{A}|B$ characterizing the quasiband below (states $14$ to $20$). 
Now, as shown in \cref{fig:fibo55Modified}\,(b), if we \emph{complete} the resonator $|A|B$ into $B|A|B$ by attaching a $B$ site to the left end of the chain, then the edge can accommodate the $b$-LRM $B|\overline{A}|B$ instead of the $e$-LRM $|\overline{A}|B$.
Consequently, the edge mode is replaced by a bulk mode of the quasiband.
In other words, the edge state is ``absorbed'' into a quasiband by converting the $e$-LRM of the former to the $b$-LRM of the latter through a structural modification at the edge.
This intuitive procedure can be applied similarly for the other two pronounced edge states (green and purple circles) in \cref{fig:fibo55Modified}\,(a), by completing the corresponding edge local resonator into a bulk one.
Thus, the selective control of a specific edge state is possible.

Let us note, however, that in most cases such a selective annihilation of one edge state leads to the creation of one (or more) other edge state(s) located elsewhere in the spectrum, as a result of the edge modification.
For example, the left edge of the modified chain in \cref{fig:fibo55Modified}\,(b) features the resonator $|BAB|A$, which is a truncated version of $A|BAB|A$ hosting the $b$-LRM $A|\overline{B}A\overline{B}|A$, thus yielding a new gap-edge state (not shown).

Interestingly, in special cases this issue can be overcome by \emph{exploiting the local symmetry of bulk resonators}, as we now explain using the example shown in \cref{fig:fibo55Modified}\,(c).
Here, an $A$ site is attached to the right edge, which formerly hosted the $e$-LRMs $B|\overline{A}B\overline{A}|$ and $B|\underline{A}B\overline{A}|$ (cf. \cref{fig:fibo55Modified}\,(a)), corresponding to the edge states $\phi^{13}$ (green) and $\phi^{22}$ (purple), respectively.
In the modified chain, the right edge features a local resonator $B|AA|$.
The key point now is that this resonator supports two LRMs, $B|\overline{A}\underline{A}|$ and $B|\overline{A}\overline{A}|$, which are degenerate to the b-LRMs $B|\overline{A}\underline{A}B\overline{A}\underline{A}|B$ and $B|\overline{A}\overline{A}B\underline{A}\underline{A}|B$, respectively, due to the reflection symmetry of the underlying resonator $B|AABAA|B$.
This symmetry-induced degeneracy is shown rigorously in \cref{appendix:symmetricMatrices}.
As a result of their degeneracy, the respective e- and $b$-LRMs can combine linearly to compose quasiband states, as seen in \cref{fig:fibo55Modified}\,(c); see states in first and third quasiband from bottom with marked edge resonators.
The same procedure can be performed on the left edge by removing an $A$ site from it, leaving the edge resonator $|B|A$ hosting $|\overline{B}|A$ which is degenerate to $A|\underline{B}A\overline{B}|A$ (see state in top quasiband with orange marked left edge resonator in \cref{fig:fibo55Modified}\,(c)).
Note that both (right and left) edge modifications above are consistent with the Fibonacci order:
The resulting chain is simply obtained from the former one by a single-site shift to the right along the infinite Fibonacci chain.
We thus have a case of finite Fibonacci chain with \emph{no edge-localized gap states}. 

From the above it is clear that edge states in binary quasiperiodic chains can now be rigorously understood and manipulated within the framework of local resonators.
Structural creation and annihilation represents a first fundamental step in edge state control.
Indeed, once an edge state is established, its energetic position within a gap can further be tuned by allowing for non-binary (freely varying) potentials at the edges, while leaving the quasibands intact.

\subsection{Behavior at low contrast} \label{sec:lowContrastBehavior}

\begin{figure}[t]
\centering
\includegraphics[max size={.99\columnwidth}{0.7\textheight}]{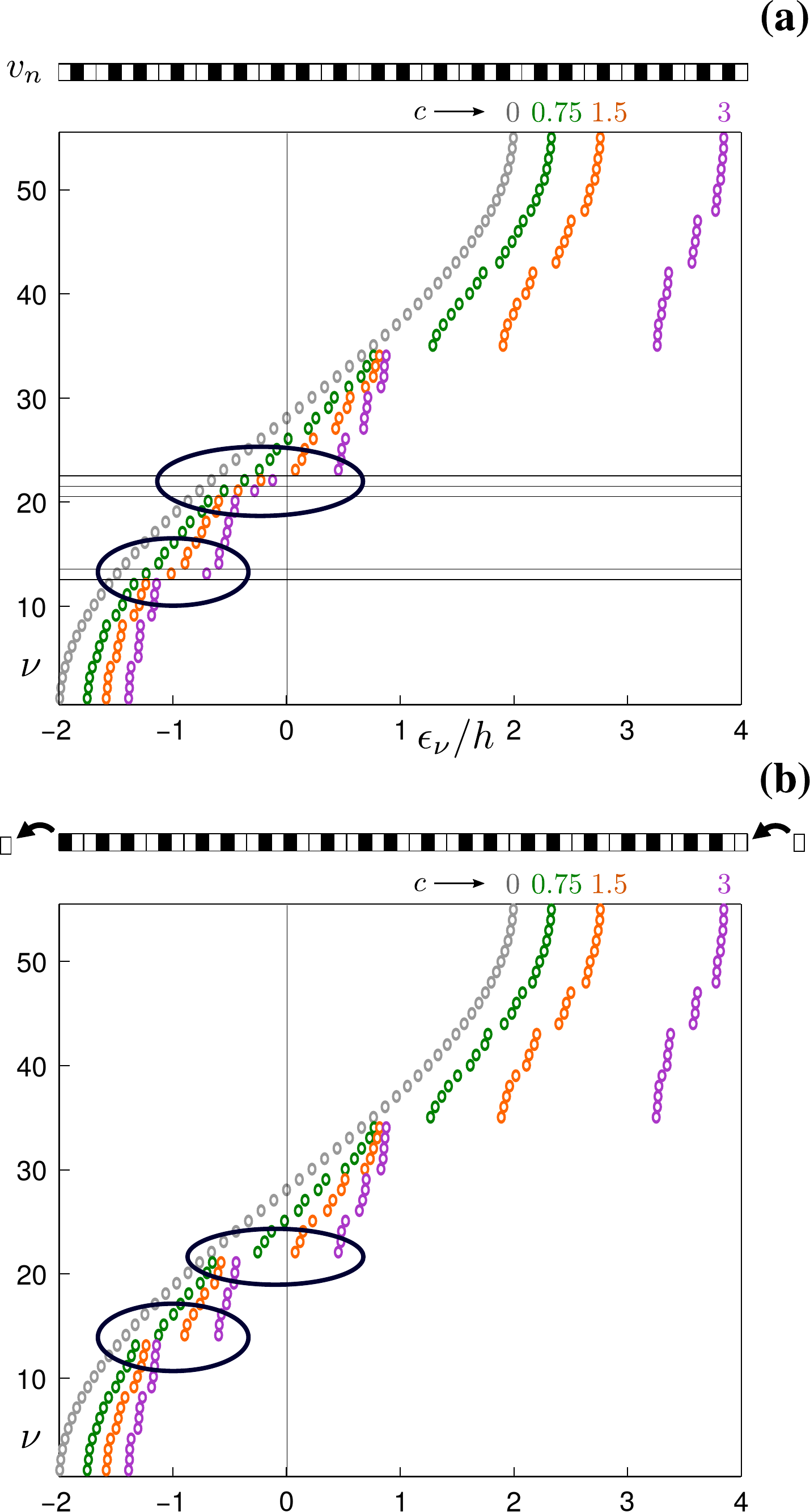}
\caption{Eigenvalue spectrum for different values of the contrast $c$ 
\textbf{(a)} for the original Fibonacci chain of \cref{fig:fibo55StateMap}\,(a), with the three gap modes indicated by horizontal stripes, and \textbf{(b)} for the modified chain of \cref{fig:fibo55Modified}\,(c). The ellipses highlight the removal of gap modes by the modification, for all contrast levels.}
\label{fig:fibo55ContrastVariation}
\end{figure}

The local understanding and controllability of edge states at high contrast levels raises the question if these properties are retained also at lower contrast.
To address this, in \cref{fig:fibo55ContrastVariation}\,(a) we show the eigenvalue spectrum of the original Fibonacci chain studied previously [\cref{fig:fibo55StateMap}\,(a)] for varying contrast $c$. 
As we see, gap states (localization on edges not shown here) are clearly distinguished for all contrast levels. 
Figure \ref{fig:fibo55ContrastVariation}\,(b) shows the spectrum of the modified Fibonacci chain of \cref{fig:fibo55Modified}\,(c)---where all edge states were annihilated at high contrast---for the same contrast values. 
Also here the structure of the spectrum is retained with varying $c$.
In particular, a real-space analysis (not shown here) confirms that all quasiband states in the modified chain remain extended in the bulk for varying $c$.
The effect of lowering the contrast is merely a reduction in the fragmentation of the eigenstate profiles which become more smeared out into regions between the LRMs defined at high contrast.

This finding indicates that the mechanism of edge state formation via truncated local resonators based on an analysis at high contrast remains valid also for lower contrast, though ``hidden'' due to the spatial smearing of the states.
In other words, the contrast parameter can be used as an intermediate tool to manipulate edge states in binary aperiodic model chains:
It is first ramped up to reveal the eigenstate structure in terms of LRMs subject to modifications, and then ramped down again with the bulk/edge-state separation retained.

\section{Application to non-quasiperiodic chains} \label{sec:nonQuasiperiodic}
Featuring a point-like spatial Fourier spectrum (rendering it, by definition, a quasicrystal\cite{Maciaroleaperiodicorder2006,DalNegro2012LPR6178DeterministicAperiodicNanostructuresPhotonics}), the Fibonacci chain studied above represents the class of lowest structural complexity when departing from periodicity towards disorder, as mentioned in \cref{sec:Introduction}.
The question naturally arises whether the local resonator framework developed in \cref{sec:edgeModes}, distinguishing edge from bulk states via LMRs, applies also to other classes of aperiodic chains.
In the following, we will demonstrate the generality of our approach by applying it to two cases of qualitatively different structural character, the Thue-Morse and Rudin-Shapiro chains.
We thereby essentially go through the same analysis steps as in \cref{sec:edgeModes}---identification of LRMs, edge state control, and low contrast behavior---and assess the particularities of each structural case.

\subsection{Singular continuous Fourier spectrum: Thue-Morse chain} \label{sec:thue}

A well-studied case of aperiodic order which is not quasiperiodic is the Thue-Morse sequence \cite{Maciaroleaperiodicorder2006}, produced by the inflation rule $A\rightarrow AB, B\rightarrow BA$ yielding $T = ABBABAABBAAB\cdots$. 
Its Fourier spectrum is singular continuous, and from this viewpoint it is considered \cite{kolarGeneralizedThueMorseChains1991} closer to the disorder limit (with absolutely continuous spectrum \cite{KroonAbsencelocalizationmodel2004}) than quasiperiodic order (with point-like spectrum). 
On the other hand, a subset of eigenstates of the Thue-Morse chain strongly resemble those of periodic chains \cite{riklundThuemorseAperiodicCrystal1987}.
The eigenstates of a $N = 144$-site Thue-Morse chain \footnote{Note that we have not used here a generation of the Thue-Morse sequence (of length \cite{MaciaBarber2008AperiodicStructuresCondensedMatter} $2^g$ for the $g$-th generation), but the same length $N = 144$ as for the Fibonacci chain.This is done for the sake of comparison but also to avoid certain symmetries \cite{noguezPropertiesThueMorseChain2001} of the Thue-Morse sequence generations (spatial mirror symmetry for odd $g$ and spectral mirror symmetry around $E = 0$ for even $g$), in favor of the generality of the analysis.} are shown in \cref{fig:thue144StateMap}.
Indeed, while some bulk states are more strongly localized into subdomains than in the Fibonacci chain for equal contrast $c = 1.5$ (compare to \cref{fig:fibo144StateMap}), others are more extended along the chain.
As we see in \cref{fig:thue144StateMap}, in spite of the quasiband structure being more fragmented, there occur well distinguishable states within gaps which are localized on one of the chain edges.

\begin{figure}[t!]
\centering
\includegraphics[max size={\columnwidth}{0.9\textheight}]{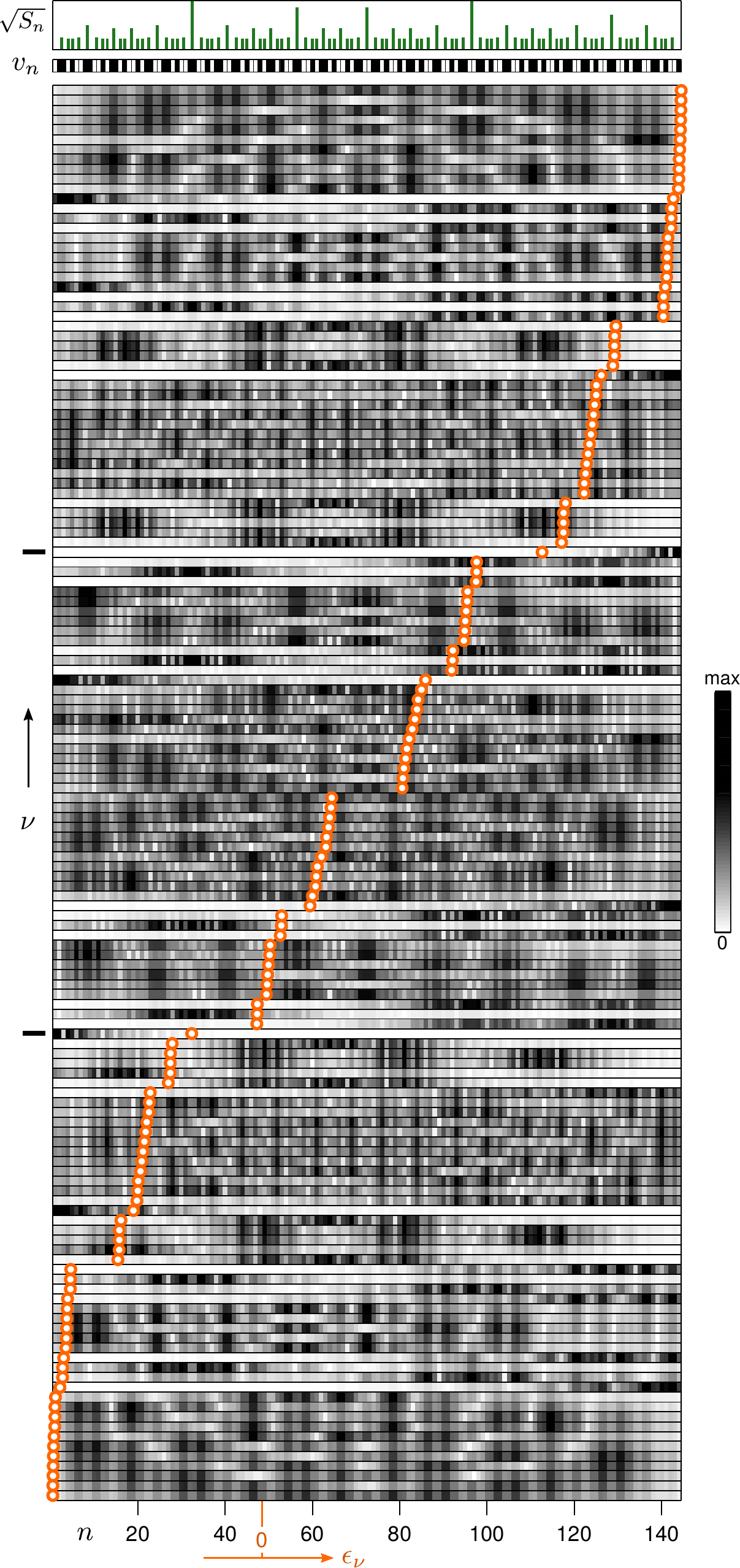}
\caption{
Like in \cref{fig:fibo144StateMap} but for a $N=144$-site Thue-Morse chain.}
\label{fig:thue144StateMap}
\end{figure}

\begin{figure}[ht!]
\centering
\includegraphics[max size={\columnwidth}{0.7\textheight}]{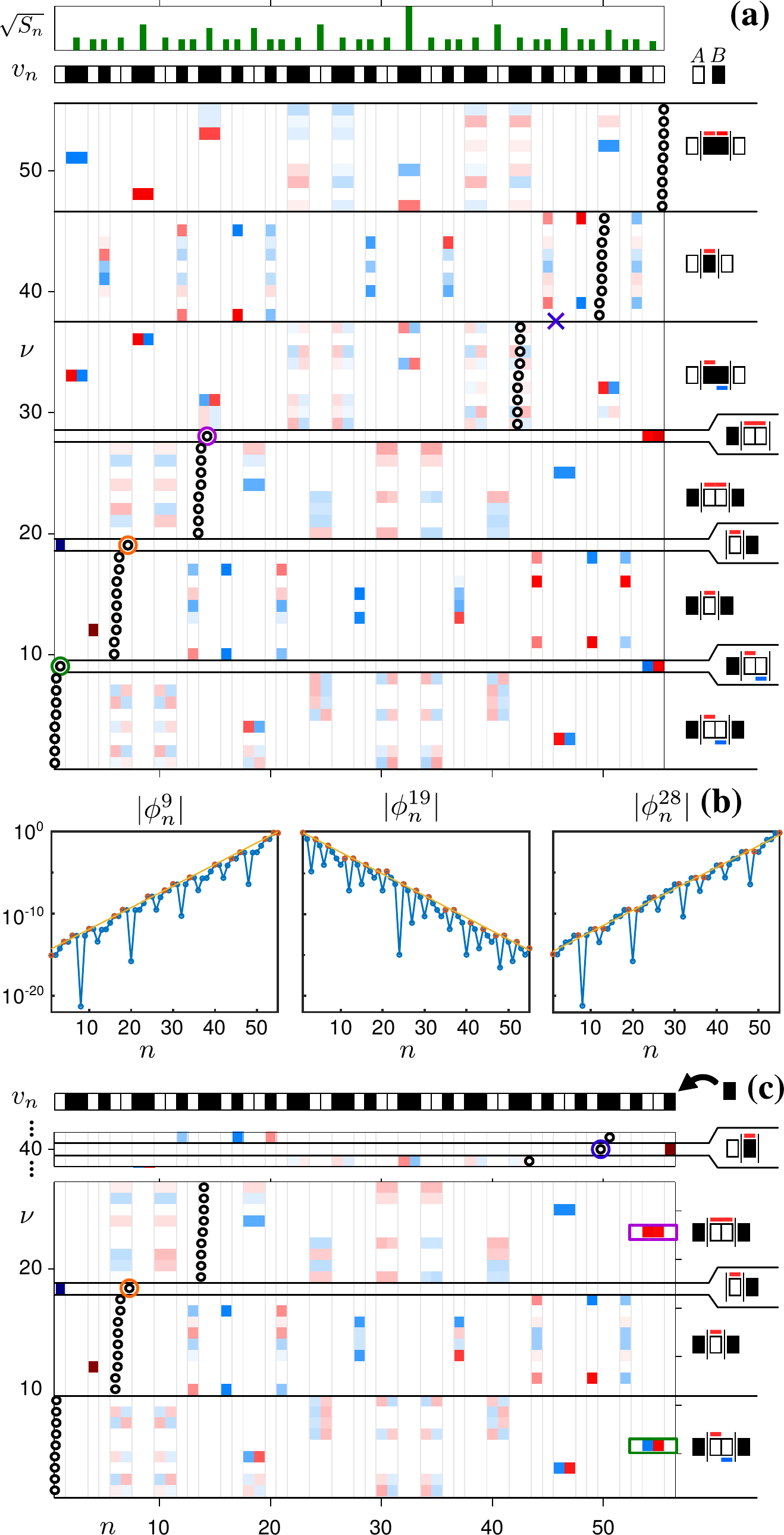}
\caption{\textbf{(a)} Like in \cref{fig:fibo55StateMap}\,(a) but for a $N=55$-site Thue-Morse chain, with three edge states marked by colored circles. 
\textbf{(b)} Absolute values of amplitudes of the three edge states $\phi^{9}$, $\phi^{19}$, $\phi^{28}$ with localization lengths $0.61$, $0.61$ and $0.6$, obtained by fitting a line (orange) to local maxima (orange dots) on a logarithmic scale.
\textbf{(c)} Absorpion of the two right edge states into quasibands (green and purple rectangles) and creation of a new right edge state (blue circle, lying between quasibands as indicated by $\times$ in (a)) by attaching a $B$ site to the right chain end, as explained in the text. Color coding of subfigures (a) and (c) is as in \cref{fig:fibo55StateMap}.
}
\label{fig:thue55StateMap}
\end{figure}

\emph{\stepOne}. 
Like in \cref{sec:edgeModes}, we consider a smaller chain of $N = 55$ sites to visually facilitate the detailed spatial analysis.
Its eigenstates are shown in \cref{fig:thue55StateMap} (a) for contrast $c = 6$, together with the distribution of local reflection symmetries in the chain (top).
As we see, the bulk state profiles are fragmented in a well-defined manner for different quasibands:
Like in the Fibonacci case, each bulk state is composed of copies of a $b$-LRM characterizing the corresponding quasiband, as indicated schematically on the right of the figure. 
In contrast, the three occurring prominent edge states (marked by colored circles) consist of non-repeated $e$-LRMs at one of the chain ends.
Like before, the local resonators underlying the $e$-LRMs can be identified as truncated local resonators underlying the $b$-LRMs.
This demonstrates that our LRM-based framework for the formation of edge states applies also for this class of aperiodic order.
Notably, also here the $b$-LRMs have definite local parity under reflection, and are present in the eigenstates following the local symmetry axes shown in the bar plot (top of \cref{fig:thue55StateMap}\,(a)). 
This is indeed predicted by the perturbation theory of \cref{app:perturbationTheory}.
We thus see that also for the Thue-Morse chain its local symmetries essentially provide the regions of localization of the eigenstate fragments at high contrast.

\emph{\stepTwo}.
The original Thue-Morse chain contained three edge states, which were exponentially localized\footnote{\matlabFootnote{}}\addtocounter{footnote}{-1}\addtocounter{Hfootnote}{-1} as shown in \cref{fig:thue55StateMap} (b).
Edge states in Thue-Morse chains were also demonstrated very recently in terms of the eigenmodes of full vectorial Green matrices \cite{WangEdgemodesscattering2018}, albeit localized according to a power-law.
Now, following the same principle as in the Fibonacci case, \cref{fig:thue55StateMap}\,(c) shows how two edge states (marked by green and orange circles in \cref{fig:thue55StateMap}\,(a)) are annihilated by attaching a $B$ site to the right end of the original chain.
Indeed, those edge states were localized on the truncated resonator $B|AA|$ which is completed to $B|AA|B$ and can thus host the $b$-LRMs $B|\overline{A}\overline{A}|B$ and $B|\overline{A}\underline{A}|B$, so that the edge states are ``absorbed'' into the corresponding quasibands.
However, the right edge of the modified chain now features a new edge state (marked by blue circle) with resonator mode $A|\overline{B}$ (its previous absence is indicated by a $\times$ in \cref{fig:thue55StateMap}\,(a)).
It lies, energetically, in the gap just below the quasiband with $b$-LRM $A|\overline{B}|A$. 
Note that, as expected from our real-space local resonator picture, the left edge state (orange circle) remains unaffected by the present modification on the right edge of the chain, since it is localized on the opposite edge.

\emph{\stepThree}. 
Finally, we investigate how edge states and quasibands behave for lower contrast in the Thue-Morse chain.
\Cref{fig:thue55ContrastVariation}\,(a) shows the spectrum of the chain of \cref{fig:thue55StateMap} (a) for varying contrast, starting from $c = 3$.
As we see, the three edge modes in the spectral gaps are clearly visible also at lower contrast levels.
The spectrum of the modified chain (with right-attached $B$ site) for varying $c$ is shown in \cref{fig:thue55ContrastVariation} (b).
As is highlighted by the black circle and the ellipse, the two former edge states are absorbed into the neighboring quasibands (as shown in \cref{fig:thue55StateMap}\,(a)) for all considered contrast levels. 
Also, the left edge state as well as the modification-induced right edge state [orange and blue in \cref{fig:thue55StateMap}\,(c), respectively] remain in their gaps as the contrast is varied.
Overall therefore, the impact of the modifications persists at lower contrast levels.

\begin{figure}[t]
\centering
\includegraphics[max size={.99\columnwidth}{0.6\textheight}]{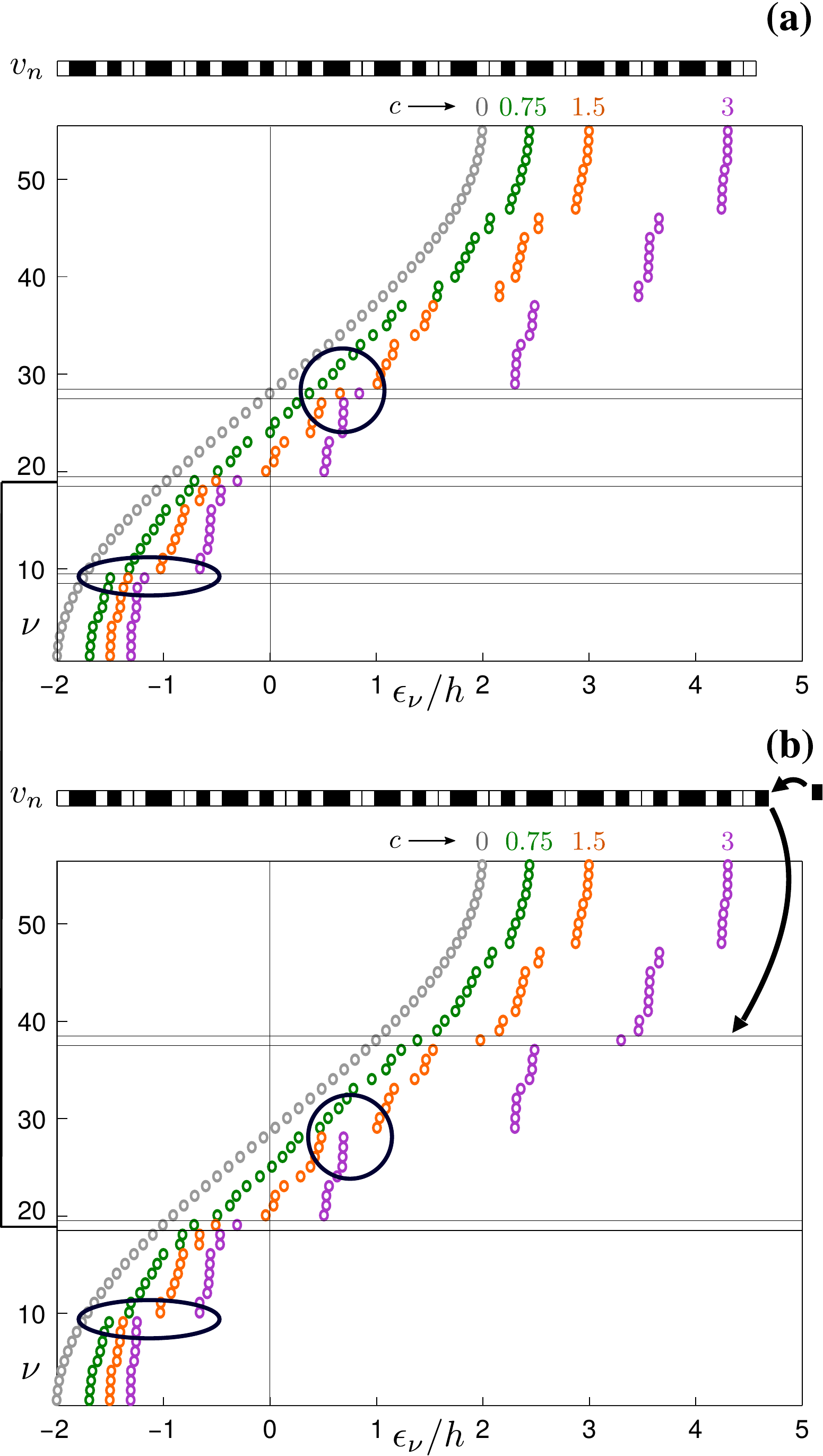}
\caption{
Eigenvalue spectrum for different values of the contrast $c$ 
\textbf{(a)} for the original Thue-Morse chain of \cref{fig:thue55StateMap}\,(a), with gap modes indicated by horizontal stripes, and \textbf{(b)} for the modified chain of \cref{fig:thue55StateMap}\,(c). The black circle and the ellipse highlight the removal of selected gap modes by the modification, for all contrast levels.}
\label{fig:thue55ContrastVariation}
\end{figure}

\subsection{Absolutely continuous Fourier spectrum: Rudin-Shapiro chain} \label{sec:rudin}

\begin{figure}[ht!]
\centering
\includegraphics[max size={\columnwidth}{0.85\textheight}]{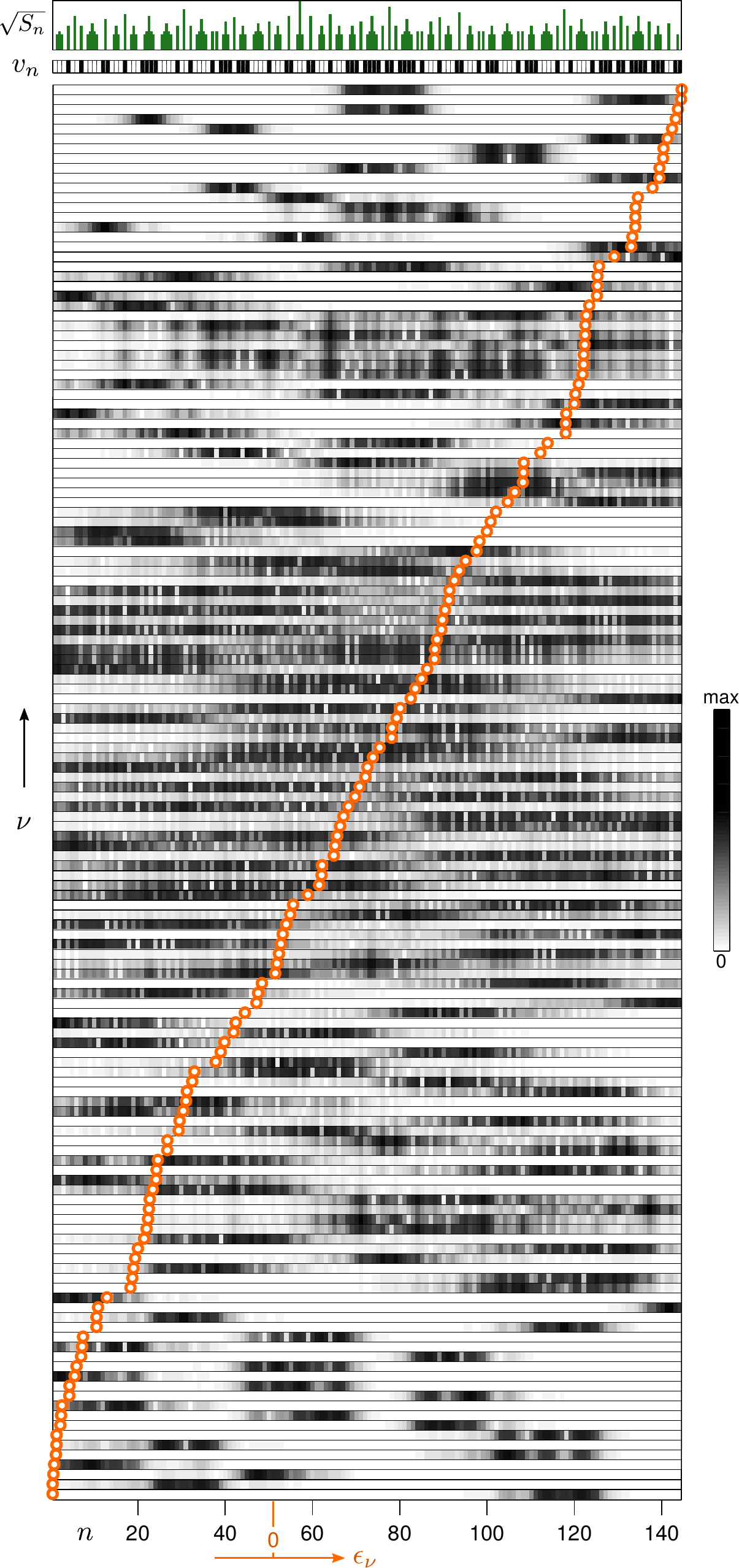}
\caption{
Like in \cref{fig:fibo144StateMap} but for a $N=144$-site Rudin-Shapiro chain.}
\label{fig:Rudin144StateMap}
\end{figure}

\begin{figure}[h!]
	\centering
	\includegraphics[max size={\columnwidth}{0.7\textheight}]{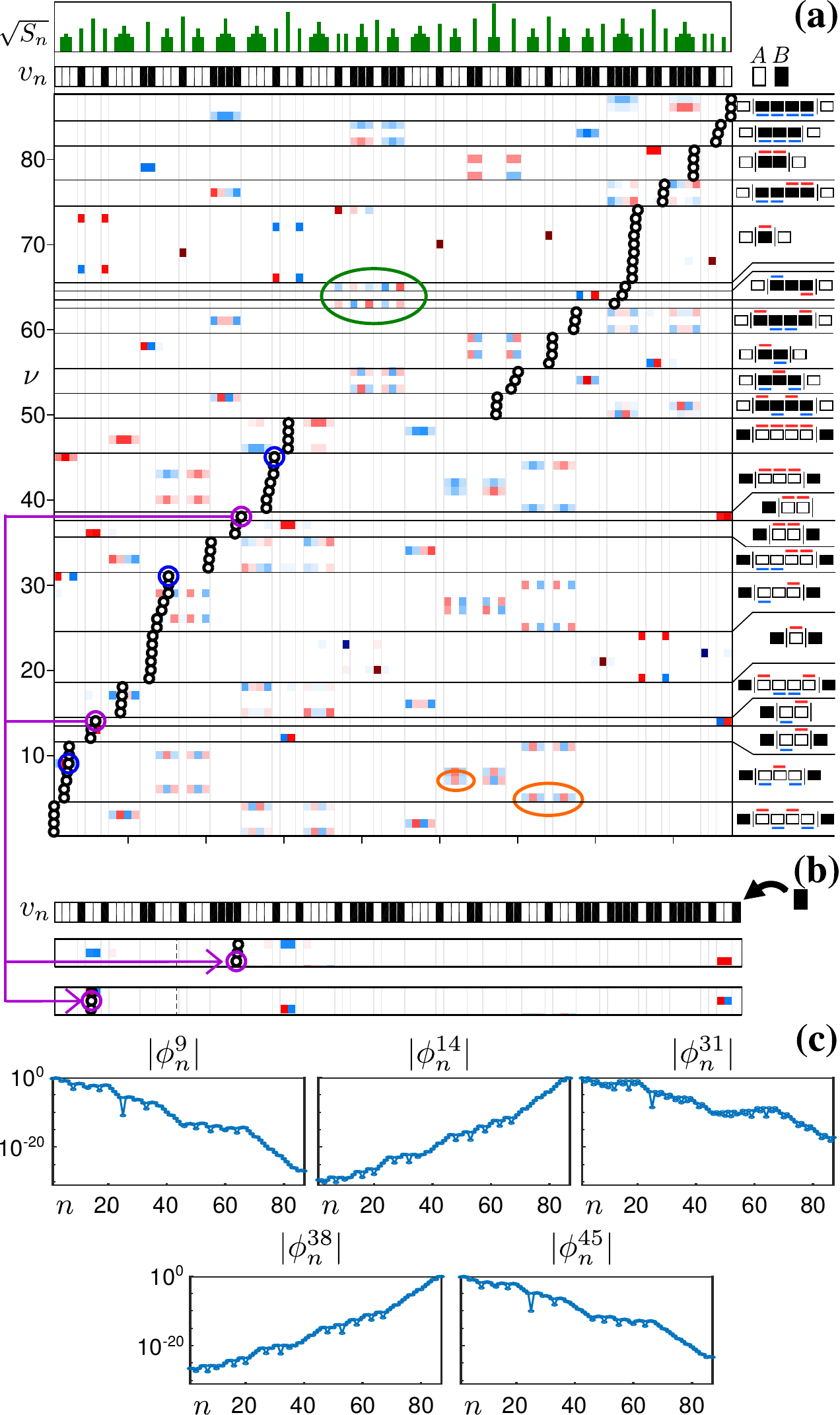}
	\caption{\textbf{(a)} $N=87$ site binary chain corresponding to a truncated Rudin-Shapiro sequence at contrast $c = |v/h| = 6$ (hopping $h=0.1$). To the right, the grouping of eigenstates into resonator modes as explained in the text is shown. To simplify the figure, resonator modes are only shown explicitly if they are shared by at least two states. The two states marked by a green ellipse localize on non-locally symmetric structures. The two states marked by orange ellipses are examples for states with different resonator modes but nearly equal energy, as explained in the text. \textbf{(b)}
	The result of an extension of the chain by adding a $B$ to the right. Due to this modification, the resonator $A|BB|$ on the right edge is completed, and the purple marked states in (a) are energetically shifted towards the corresponding states localizing on the $A|BB|A$ resonator located near the right edge. Color coding of subfigures (a) and (b) is as in \cref{fig:fibo55StateMap}. \textbf{(c)} Absolute values of amplitudes of the edge states $\phi^{9}$, $\phi^{14}$, $\phi^{31}$, $\phi^{38}$ and $\phi^{45}$.}
	\label{fig:Rudin55StateMap}
\end{figure}

Taking a step towards higher structural complexity, we finally investigate the case of a Rudin-Shapiro chain in terms of our local resonator framework.
The Rudin-Shapiro sequence \cite{KroonLocalizationdelocalizationaperiodicsystems2002} $R$ is obtained by the inflation rule
$AA\rightarrow AAAB$, $AB\rightarrow AABA$, $BA\rightarrow BBAB$, $BB\rightarrow BBBA$, yielding $R = AAABAABAAAABBBAB\cdots$ for an initial seed $AA$.
Its Fourier spectrum is absolutely continuous, a property shared with completely disordered chains \cite{KroonAbsencelocalizationmodel2004}.
Further, there are indications that the tight-binding Rudin-Shapiro chain has both exponentially and weaker-than-exponentially localized eigenstates \cite{MaciaNatureElectronicWave2014,DuleaUnusualscalingspectrum1993,DuleaTracemapinvariantzeroenergy1992,duleaLocalizationElectronsElectromagnetic1992}, while even extended ones have been shown to exist at low contrast.
The different character of the Rudin-Shapiro states compared to the Fibonacci or Thue-Morse chain can be anticipated from the eigenstate map shown in \cref{fig:Rudin144StateMap}.
As we see, at this low contrast ($c = 1.5$) there is now no clear distinction between bulk and edge states. 
Moreover, no clear energetic clustering into well-defined quasibands is present. 
Note also that the distribution of local reflection symmetries along the chain (see top of figure) is much less structured than in the Fibonacci or Thue-Morse chains (cf. top of \cref{fig:fibo144StateMap,fig:thue144StateMap}), with overall smaller symmetry domains present.
At the same time, there is clustering of symmetry axes with gaps in between, caused by the occurrence of larger contiguous blocks of single type (up to four $A$ or $B$ sites in a row) along the sequence.
In the following we show that there is still a strong link of the eigenstates and spectral features of the Rudin-Shapiro to the presence of locally symmetric resonators.

\emph{\stepOne}.
For the high-contrast analysis, we consider a Rudin-Shapiro chain of $N = 87$ sites.
The size is now chosen slightly larger in order to better reflect the structural properties of the Rudin-Shapiro sequence.
Indeed, in accordance with its higher complexity, a given substructure will here repeat at relatively larger distances along the sequence.
It may thus occur only once in a too short chain, thereby obscuring its long-range order.
\Cref{fig:Rudin55StateMap}\,(a) shows the eigenstate map of the considered chain at contrast $c = 6$. 
We see that also here the eigenstates fragment onto locally symmetric substructures, and are again composed of $b$-LRMs corresponding to clustered eigenvalue quasibands, as shown on the right. 
The difference is now that there are many more different identified $b$-LRMs compared to the Fibonacci and Thue-Morse chains.
This is because the increased number of contiguous block sizes allows for a \emph{higher diversity of local resonator substructures}, with larger resonators additionally hosting a larger number of different LRMs each.
In turn, there is a higher possibility that different $b$-LRMs have (nearly) the same energy, since the different resonators may have partially overlapping individual eigenspectra. 
Therefore, it may now more easily occur that \emph{different} LRMs participate in the \emph{same} eigenstate (to which they are quasidegenerate; see \cref{appendix:submatrixTheorem}).
An example of this are the states indicated by the green ellipse in \cref{fig:Rudin55StateMap}\,(a):
Each of them consists of a $A|\overline{B}|A$ on the left and two $A|\overline{B}\underline{B}\overline{B}|A$ on the right, consecutively overlapping by one $A$ site.
The emergence of such modes is explained in detail by means of perturbation theory in \Cref{app:perturbationTheory}.
Further, edge states appear which localize on corresponding $e$-LRMs.
Those are now, however, energetically not as clearly distinct from the clustered eigenvalues of quasibands as in the Fibonacci and Thue-Morse cases.
For example, the states marked by blue circles are localized on the left edge, but are composed of the $e$-LRMs (from top to bottom) $|\overline{AAA}|B$, $|\overline{A}\underline{A}\overline{A}|B$, $|\overline{A}A\underline{A}|B$, which are nearly degenerate to the $b$-LRMs of the corresponding quasibands (see right side of figure).
Nevertheless, there are also well-distinguished edge states lying in gaps (though close to gap edges) marked by purple circles.

\emph{\stepTwo}.
Contrary to the Fibonacci and Thue-Morse cases, the edge states are not exactly exponentially localized, but have different localization lengths in different sections, as shown\footnotemark in \cref{fig:Rudin55StateMap} (c). The amplitude of state $\phi^{31}$, though overall decaying, even rises again at around $n\approx 20$ and $n \approx 70$.
Unaffected by this different localization behavior compared to the previously treated examples, we now manipulate the two states marked by purple color which are localized on the right edge. 
These localize on the truncated resonator $B|AA|$, and their energy is different from the energy of states localized on the complete resonator $B|AA|B$ which occurs twice in the bulk.
In (b) we add a $B$ on the right edge of the chain, completing this resonator. Due to this completion, the two former gap-edge states move into the respective energy cluster (or quasiband).

\emph{\stepThree}. In \cref{fig:rudin87ContrastVariation} (a) we investigate the eigenvalues of the Rudin-Shapiro chain of \cref{fig:Rudin55StateMap} (a) for varying contrast. Compared to the case of the Fibonacci chain presented in \cref{fig:fibo55ContrastVariation} or the Thue-Morse chain presented in \cref{fig:thue55ContrastVariation}, the energetic clusters form only at high contrast values.
This already indicates that modifications to the chain done at high contrast can not directly be traced to energetic changes at low contrast as was the case for the Fibonacci and Thue-Morse chain. This can also be seen for the two edge states marked by horizontal lines in \cref{fig:rudin87ContrastVariation} (a). At high contrast, these are caused by a truncated resonator $B|AA|$ on the right edge. In \cref{fig:Rudin55StateMap} (b) we have completed this resonator, causing the two edge states to move (at high contrast) closer to the nearest eigenvalue cluster. As can be seen in \cref{fig:rudin87ContrastVariation} (b), this manipulation is only effective at high contrast. For low contrast, the eigenvalue structure is nearly unchanged compared to the original chain shown in \cref{fig:rudin87ContrastVariation} (a).

\begin{figure}[t]
	\centering
	\includegraphics[max size={.99\columnwidth}{0.7\textheight}]{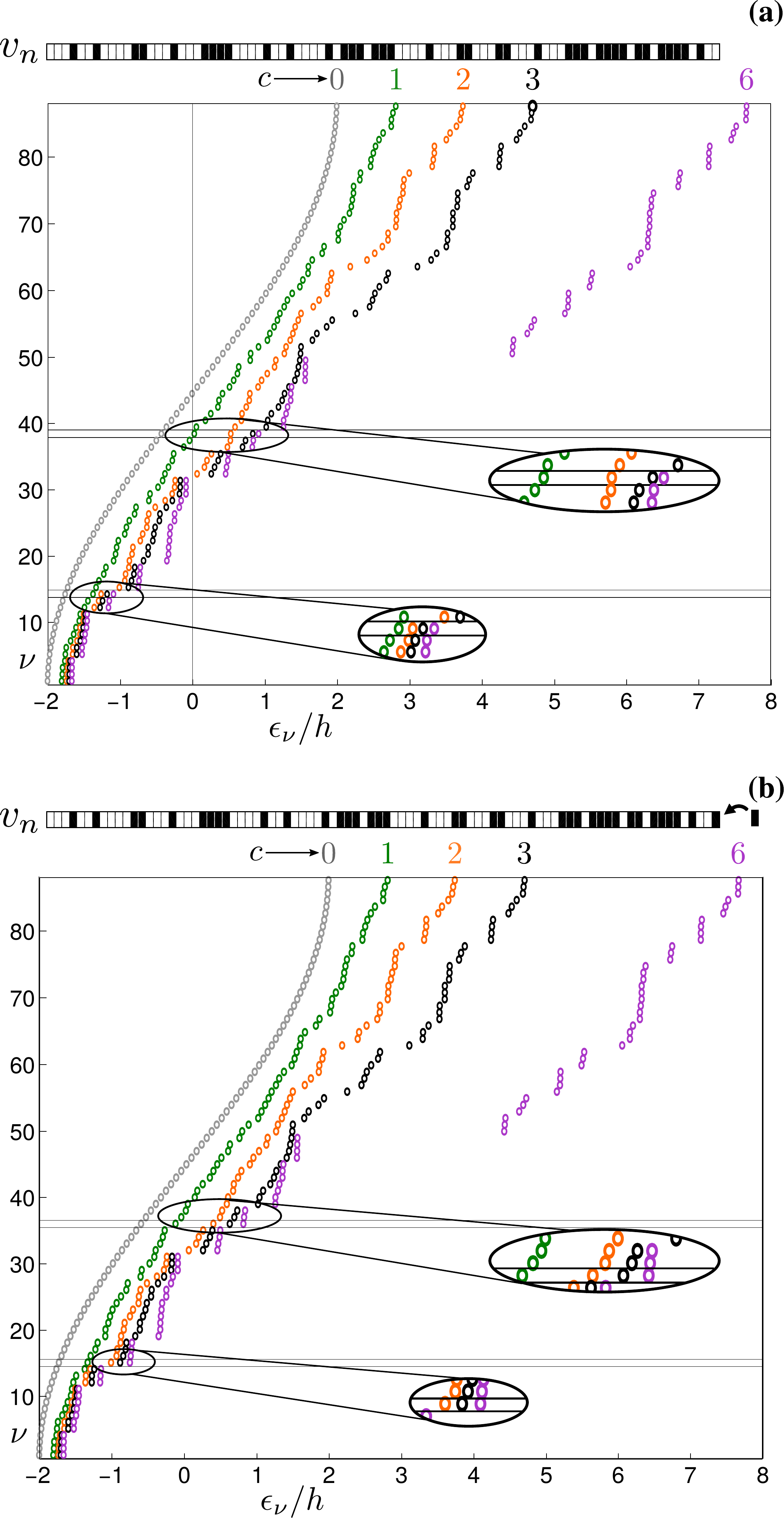}
	\caption{
		\textbf{(a)} Evolution of the eigenvalue spectrum of the Rudin-Shapiro chain shown in \cref{fig:Rudin55StateMap} (a) for various values of the contrast $c$. The two gap states $\phi^{14,38}$ are denoted by horizontal lines. \textbf{(b)} Same as (a), but now for the modified Rudin-Shapiro chain shown in \cref{fig:Rudin55StateMap} (b). For high contrast of $c = 6$, the two gap states are removed. However, they reappear, though at slightly different positions, already at a contrast $c = 3$, as shown in the insets.}
	\label{fig:rudin87ContrastVariation}
\end{figure}

\vspace{1ex}

In conclusion, we have applied our local-symmetry based resonator strategy to the Thue-Morse and the Rudin-Shapiro chain.
The results show that our approach can be used to explain and control gap-edge states of the Thue-Morse chain.
At high contrast the gap-edge states of the Rudin-Shapiro chain are likewise explained. However, our approach can not be used to make qualitative predictions at low contrast.

\section{Applicability and relation to other approaches} \label{sec:commentsOnGenerality}
Let us briefly comment on the limitations of applicability of the developed framework and its connection to similar approaches in the literature.
The presented methodology essentially relies on the fragmentation of eigenstates at high contrast and can thus only be applied onto chains featuring such a behavior of eigenstates. A perturbation theoretical treatment of binary tight-binding chains which serves as basis for our methodology, see \Cref{app:perturbationTheory}, indicates that at high contrast the fragmentation of eigenstates is indeed the generic case. However, the necessary conditions for this behavior still need to be determined in order to clarify the range of applicability.

The connection between local resonators and quasibands in quasiperiodic setups has been commented on in Ref. \onlinecite{ZijlstraExistencelocalizationsurface1999a,Macia2017PSSb2541700078ClusteringResonanceEffectsElectronic,BandresTopologicalPhotonicQuasicrystals2016}. For the Thue-Morse sequence, a similar analysis has been achieved in Ref. \onlinecite{riklundThuemorseAperiodicCrystal1987}.
However, to the best of our knowledge, there is no systematic framework bringing together the three concepts of LRMs, quasibands and edge states into a unified context.
An approach related to ours is the renormalization group flow analysis. For the tight-binding chains, this method aims at understanding the energetic behavior of a chain through a series of size reductions\cite{niuSpectralSplittingWavefunction1990}. At each step, the size of the system is decreased, and the behavior of the decreased one is linked to the bigger one by a renormalization procedure, usually done in terms of perturbation theory. The renormalization group flow is a powerful method, and has been successfully used to explain the fractal nature of the Fibonacci spectrum \cite{kaluginElectronSpectrumOnedimensional1986,piechonAnalyticalResultsScaling1995,liuBranchingRulesEnergy1991,niuSpectralSplittingWavefunction1990,niuRenormalizationGroupStudyOneDimensional1986}.
However, it needs to be tailored to the system of interest, and as stated in Ref. \onlinecite{niuRenormalizationGroupStudyOneDimensional1986}, finding an appropriate renormalization group flow for a general quasiperiodic chain is not easy. This stands in contrast to the very general method proposed in this work, which was shown to be applicable to a broad range of different setups.

\section{Conclusions and outlook} \label{sec:conclusions}

We have presented a systematic approach to the analysis of aperiodic binary tight-binding chains regarded as a combination of different resonator-like subsystems rather than a single bulk unit. For low inter-site coupling, each eigenstate is seen to be composed of spatially non-overlapping local resonator modes of these resonator structures.
This viewpoint, supported by a rigorous perturbation theoretical treatment, allows for an intuitive explanation of the emergence of both quasibands and gap-edge states in such chains.
We demonstrate the power of our approach by applying it to Fibonacci, Thue-Morse and Rudin-Shapiro chains and show how gap-edge states occurring in these chains can be manipulated.

A repeating motif in our analysis of eigenstates at high contrast is the fact that most resonator modes share the local symmetries of the underlying systems. This strong impact of local symmetries is remarkable, especially as it is hidden at lower contrast levels by a substantial background in the eigenstate profiles. In this work we have given an explanation for this finding at high contrast, and we believe that the study of local symmetries in complex setups is a very promising field with rich perspectives and potential applications.
The recently established framework of local symmetries \cite{Kalozoumis2015AP362684InvariantCurrentsScatteringLocally,Kalozoumis2014PRL11350403InvariantsBrokenDiscreteSymmetries,Kalozoumis2013PRA8833857LocalSymmetriesPerfectTransmission,Morfonios2017AP385623NonlocalDiscreteContinuityInvariant,Morfonios2017AP385623NonlocalDiscreteContinuityInvariant,Rontgen2017AP380135NonlocalCurrentsStructureEigenstates,Zampetakis2016JPA49195304InvariantCurrentApproachWave} provides dedicated tools for this purpose, and extensions of it are of immediate relevance.
In this line, our work may enable the local-symmetry assisted design of novel optical devices that support desired quasiband structures and strongly localized edge states at prescribed energies, offering exciting opportunities to control light-matter coupling in complex aperiodic environments.

\section{Acknowledgments}

M.R. gratefully acknowledges financial support by the `Stiftung der deutschen Wirtschaft' in the framework of a scholarship. L.D.N. gratefully acknowledges the sponsorship of the Army Research Laboratory under
Cooperative Agreement Number W911NF-12-2-0023. The views and conclusions contained in
this document are those of the authors and should not be interpreted as representing the official
policies, either expressed or implied, of the Army Research Laboratory or the U.S. Government.
The U.S. Government is authorized to reproduce and distribute reprints for Government purposes
notwithstanding any copyright notation herein.

\appendix

\section{Discrete resonators} \label{appendix:discreteResonators}

\begin{figure}[t]
	\centering
	\includegraphics[max size={1\columnwidth}{0.8\textheight}]{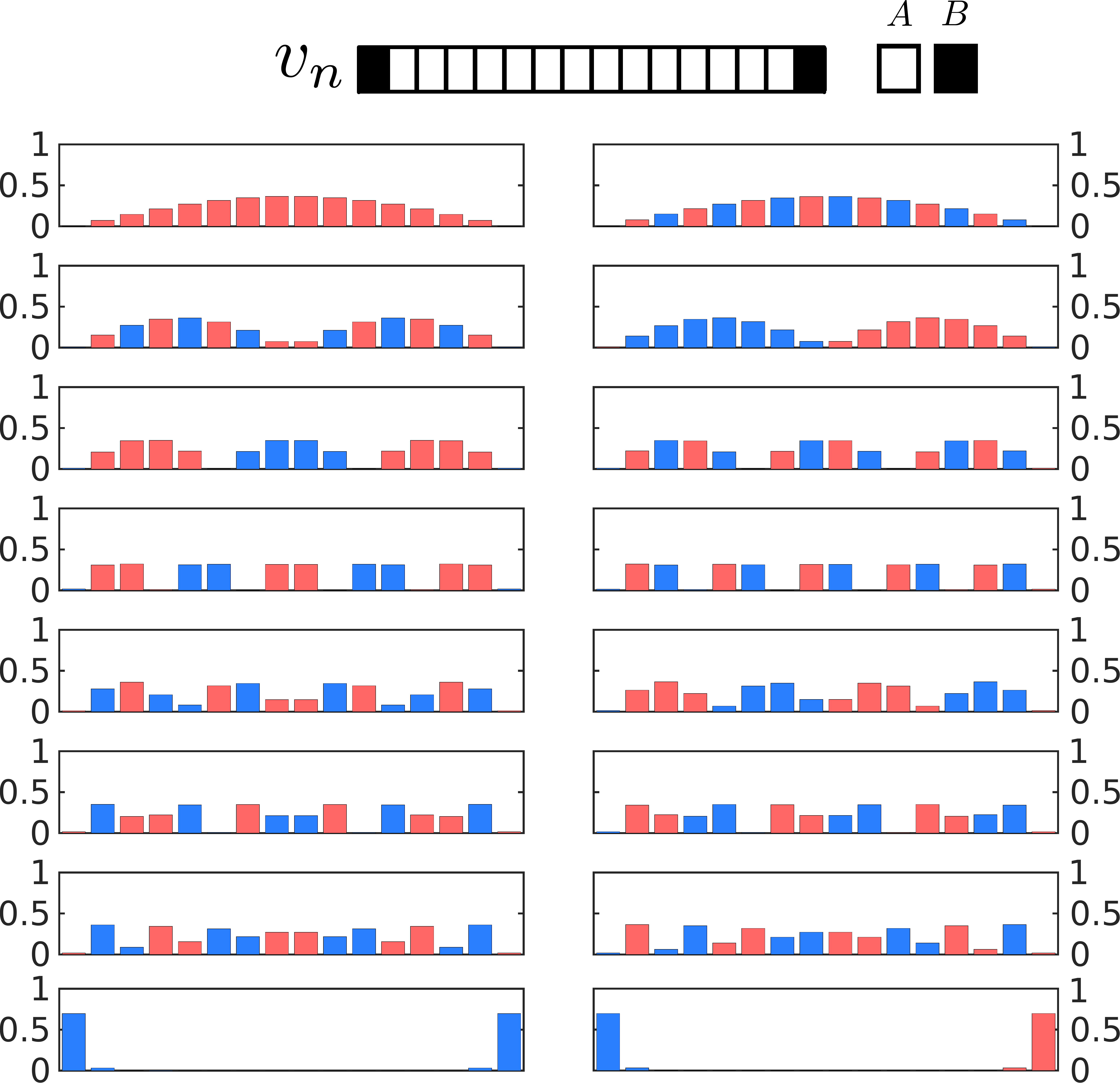}
	\caption{All $16$ eigenstates of the chain (depicted above) $BA\ldots{}B$ with $14$ $A$-sites at a contrast of $c = 20$. All but the eigenstates in the last row localize on the $A$-sites.
		}
	\label{fig:resonatorOverview}
\end{figure}

The aim of this appendix is to justify viewing substructures embedded in a larger binary aperiodic lattice as local resonators.
To this end, we investigate the behavior of the simplest case of such a structure, $BAB$, in more detail.
Its Hamiltonian is
\begin{equation} \label{eq:resonatorHamiltonian}
H = \begin{pmatrix}
v_{B} & h & 0 \\
h & v_{A} & h \\
0 & h & v_{B}
\end{pmatrix},
\end{equation}
with the (unnormalized) eigenstates
\begin{equation}
\phi^{1} = \begin{pmatrix}
-1 \\ 0 \\ 1
\end{pmatrix}, \phi^{2,3} = \begin{pmatrix}
1 \\ \frac{-\delta \pm \sqrt{8 + \delta^2}}{2} \\ 1
\end{pmatrix}.
\end{equation}
where $\delta = (v_{A} - v_{B})/h$, with $c = |\delta|$.
For high contrast $c$, $\phi^{3} \approx (1,-\delta,1)^{T}$ localizes on the central site.
The idea now is to view $BAB$ at high contrast as a resonator, where the site $A$ effectively plays the role of a cavity, while the outer sites $B$ play the role of cavity walls.
The resemblance to a resonator becomes clearer for a larger structure with more modes between the resonator walls, like the structure in \cref{fig:resonatorOverview}.
As one can see, all but two eigenstates extend nearly exclusively on the internal $A$ sites, and the wave-like character of these states is well recognizable. Two states exclusively localize on the outer two $B$ sites.
The setup thus acts as an extended cavity consisting of $14$ $A$ sites, with two $B$ sites playing the role of the cavity walls.
The smaller structure $BAB$ is of the same nature, albeit with a cavity of only a single site $A$.
Notationally, we will divide the actual cavity and the cavity walls of a resonator by a vertical line, writing e.g.
$B|A\ldots{}A|B$.
Similarly, we also view the ``inverse'' structure
$A|B\ldots{}B|A$
as a resonator with resonator modes of higher energy, assuming $v_{B} \gg v_{A}$.
Moreover, closely neighboring resonators of the form
\begin{equation} \label{eq:resonatorStructuresDouble}
B|A|B|A|B,\; B|AA|B|AA|B,\ldots
\end{equation}
can be seen as coupled resonators.
To indicate the composite character of such resonators, we omit the inner vertical lines, i.e., $B|ABA|B,\; B|AABAA|B,\ldots{} \;$.

\section{Symmetry argument for the absence of edge states} \label{appendix:symmetricMatrices}
Here we explain the absence of edge states in \cref{fig:fibo55Modified} (c) using the concept of local symmetry. The underlying symmetry concept is very general and not limited to the Fibonacci chain, as we demonstrate in the last paragraph of this appendix.
Let us denote an arbitrary sequence of $A$'s and $B$'s by $X$, its reverse ordered counterpart by $X^{-1}$, and by $Y$ a single site $A$ or $B$. Then
\begin{equation} \label{eq:equitableSymmetry}
\sigma([X]) \subset \sigma([X^{-1}YX])
\end{equation}
where $\sigma$ denotes the eigenvalue spectrum and $[X]$ the tridiagonal Hamiltonian representing $X$.
In words, the eigenvalue spectrum of a resonator $[X]$ is completely contained in that of the reflection-symmetric resonator $[X^{-1}YX]$.
For example, if $X = AB$ and $Y = B$, then $X^{-1} = BA$ and $\sigma([AB]) \subset \sigma([BABAB])$.

To prove the above statement, we note that the Hamiltonian $[X^{-1}YX]$ reads
\begin{equation}
	H = \begin{pmatrix}
	\lbrack X^{-1}\rbrack & C & 0\\
	C^{T} & [Y] & D \\
	0 & D^{T} & \lbrack X \rbrack
	\end{pmatrix}
\end{equation}
where $[X^{-1}],\; [X] \in \mathbb{R}^{m\times{}m}$. The matrices $C = (0,\ldots,0,h)^{T} \in \mathbb{R}^{m\times{}1}$ and $D = (h,0,\ldots,0) \in \mathbb{R}^{1\times{}m}$ connect the central site $[Y]$ to $[X]$ and $[X^{-1}]$, respectively.
Now, using the ``equitable partition theorem'' from Ref. \onlinecite{Rontgen2018PRB9735161CompactLocalizedStatesFlat}, we can transform $H$ by a similarity transform into a block-diagonal form
\begin{equation}
	H' = \begin{pmatrix}
	\lbrack X^{-1} \rbrack & \sqrt{2} C & 0\\
	\sqrt{2} C^{T} & [Y] & 0 \\
	0 & 0 & [X]
	\end{pmatrix}.
\end{equation}
The similarity transform conserves $\sigma$, and since $H'$ is block-diagonal, we have
\begin{equation}
	\sigma(H) = \sigma(H') \Rightarrow \sigma([X]) \subset \sigma(H) = \sigma([X^{-1}YX])
\end{equation}
which proves \cref{eq:equitableSymmetry}. Moreover, again using the equitable partition theorem, one can show that the eigenvalues of $[X]$ belong to eigenstates of $[X^{-1}YX]$ with negative parity with respect to the central site $Y$.

Let us now apply the above statement to \cref{fig:fibo55Modified} (c). Here, for each resonator mode at the edge, there exists one resonator mode within the bulk possessing a similar energy:
\begin{align} \label{eq:antisymmetricResonatorModes}
	\epsilon(|\overline{B}|A) &\approx \epsilon(A|\overline{B}A\underline{B}|A) \\
	\epsilon(A|\underline{B}\overline{B}|) &\approx  \epsilon(A|\underline{B}\overline{B}A\underline{B}\overline{B}|A)  \label{eq:antisymmetricResonatorModes2}\\
	\epsilon(A|\overline{BB}|) &\approx \epsilon(A|\overline{BB}A\underline{BB}|A) \label{eq:antisymmetricResonatorModes3}
\end{align}
where $\epsilon(R)$ denotes the energy of the resonator mode $R$.
In the limit of high contrast, where the resonators present in \cref{eq:antisymmetricResonatorModes,eq:antisymmetricResonatorModes2,eq:antisymmetricResonatorModes3} are disconnected from the remainder of the system, the approximations become equalities, and the edge state eigenenergies are thus ``absorbed'' into the corresponding quasiband.

In a similar manner, the energetic near-equivalence of resonator modes
\begin{align*}
		\epsilon(|\overline{AAA}|B) &\approx \epsilon(B|\overline{AAA}B\underline{AAA}|B) \\
		\epsilon(|\overline{A}\underline{A}\overline{A}|B) &\approx \epsilon(B|\overline{A}\underline{A}\overline{A}B\underline{A}\overline{A}\underline{A}|B) \\
		\epsilon(|\overline{A}A\underline{A}|B) &\approx \epsilon(B|\overline{A}A\underline{A}B\overline{A}A\underline{A}|B)
\end{align*}
at high contrast as occurring in \cref{fig:Rudin55StateMap} (a) can be explained.

\section{Perturbation theoretical treatment} \label{app:perturbationTheory}

In this section, we give an explanation for the fragmentation of eigenstates at high contrast in terms of a perturbation theoretical analysis. This will also show why the dominant entries of the eigenstates are in almost all cases obeying local symmetries.
Before we start, we note that a degenerate perturbation theoretical treatment of binary chains has been done in the past to retrieve its eigenenergies \cite{Barache1994PRB4915004ElectronicSpectraStronglyModulated}. The main focus in the following, however, lies on the behavior of eigenstates.

To apply perturbation theory, we write the Hamiltonian \cref{eq:tridiagonalHamiltonian} as
\begin{equation} \label{eq:definitionOfPerturbation}
	H = H_{0} + \lambda H_{I} \in \mathbb{R}^{N\times{} N},
\end{equation}
where $H_{0}$ solely contains the diagonal part of $H$, i.e., isolated sites, while $H_{I}$ connects them, i.e., contains the off-diagonal elements of $H$.
By means of $\lambda$, an eigenstate $\ket{\phi}^{(i)},\; 1 \le i \le N$ of $H$ as well as its energy $\en{i}{}$ are expanded as
\begin{align} \label{eq:perturbationSeriesStates}
	\pState{i}{} & = \pState{i}{0}  + \lambda \pState{i}{1}  + \lambda^2 \pState{i}{2}  + \ldots \\
	\en{i}{} &= \en{i}{0} + \lambda \en{i}{1} + \lambda^2 \en{i}{2} + \ldots \label{eq:perturbationSeriesEnergies} .
\end{align}
Inserting \cref{eq:perturbationSeriesStates,eq:perturbationSeriesEnergies} into the Schrödinger equation $H \ket{\phi^{(i)}} = \epsilon^{(i)} \ket{\phi^{(i)}}$ yields the perturbation series which is assumed to converge and thus solved order by order in $\lambda$.

At zeroth order, the perturbation series reduces to the eigenvalue equation for the unperturbed $H_{0}$. Since it is binary, the $N$ eigenstates of $H_{0}$ are highly degenerate and form two groups, satisfying
\begin{align*}
H_{0} \psiState{\alpha}{} &= v_{A} \psiState{\alpha}{}, \; 1 \le \alpha \le g_{A} \\
H_{0} \psiState{\beta}{} &= v_{B} \psiState{\beta}{}, \; g_{A} + 1 \le \beta \le g_{A} + g_{B} = N
\end{align*}
where $g_{A,B}$ denote the number of sites with potential $A,B$, respectively. The so-called ``correct'' zeroth-order states which fulfill
\begin{equation} \label{eq:definitionOfCorrectZerothOrderStates}
\pState{g}{0} = \lim\limits_{\lambda \rightarrow 0}{\pState{g}{}}, \; g = \{\alpha,\beta\}
\end{equation}
and which occur in \cref{eq:perturbationSeriesStates} and thus also in the perturbation series are linear superpositions of the $\psiState{g}{}$.
In the following, we will always denote the two sets $\{\alpha,\beta\}$ by $g$ and simple call the $\pState{g}{0}$ the zeroth-order states.

At the start of the perturbation theoretical treatment, the $\psiState{i}{},\; 1 \le i \le N$ are known, but the $\pState{i}{0}$ are usually not, and the $\pState{i}{1,2,\ldots{}}$ can not be directly be determined.
However, it can be shown\cite{Hirschfelder1974JCP601118DegenerateRSPerturbationTheory} that already the knowledge of the $\psiState{i}{}$ is sufficient to obtain a series of particular solutions to the $1,2,\ldots{},n$-th order perturbation equation, yielding the energy-corrections $\en{i}{1},\ldots{},\en{i}{n}$ as a byproduct.
Provided that the degeneracy of a given state $\pState{j}{},\; 1 \le j \le N$ is lifted at $k$-th order, then the corresponding correct-zeroth order state $\pState{j}{0}$ can be obtained by diagonalizing a $\mathbb{R}^{|j'| \times{} |j'|}$ matrix which can be derived from the $(k-1)$-th order perturbation equation\cite{Hirschfelder1974JCP601118DegenerateRSPerturbationTheory}. Here, $|j'|$ is the number of states $\pState{j'}{}$ which are degenerate with $\pState{j}{}$ up to order $k-1$. Then, at order $k+1,\ldots{},k+l$, the state correction $\pState{j}{1},\ldots{},\pState{j}{l}$ can be obtained.
Note that for the problem at hand, all degeneracies \emph{are guaranteed} to be lifted at a finite order, since the eigenvalues of tridiagonal matrices with strictly non-vanishing sub-and superdiagonals (such as the one here) are distinct\cite{Parlett1998SymmetricEigenvalueProblem} (i.e., non-degenerate).
Though all degeneracies will eventually be lifted, the order at which this happens is in general different for different states.
In many textbooks, all degeneracies are resolved already at first order, and the zeroth-order states $\pState{i}{0}$ are the ones that diagonalize the matrix $\braket{\psi^{(i)}|H_{I}|\psi^{(i)}}$ in the corresponding degenerate subspace.
This results in simple expressions for the higher-order corrections for both the states and the energy. For our binary $H_{0}$, however, degeneracies are usually resolved only at very high order, and the process becomes complex. For Fibonacci chains, all degeneracies are resolved at $4$-th order for generation $g=7$, at $5$-th order for $g=8$, at $6$-th order for $g=10$ but only in $8$-th order for $g = 12$.

In the following, we will first show the feasibility of degenerate perturbation theory by means of the Fibonacci chain, showing that for high contrast already the zeroth-order states are sufficient to explain the fragmentation of states. Next, we will show the process of determining the zeroth-order states in the first three orders, allowing for an intuitive picture of the emergence of fragmentation and locally symmetric amplitudes.
We have numerically observed convergence of the perturbation series if the contrast is larger than roughly $5$, depending on the exact chain.

\subsubsection*{Application onto the Fibonacci chain}

\begin{figure}[t]
	\centering
	\includegraphics[max size={.99\columnwidth}{0.7\textheight}]{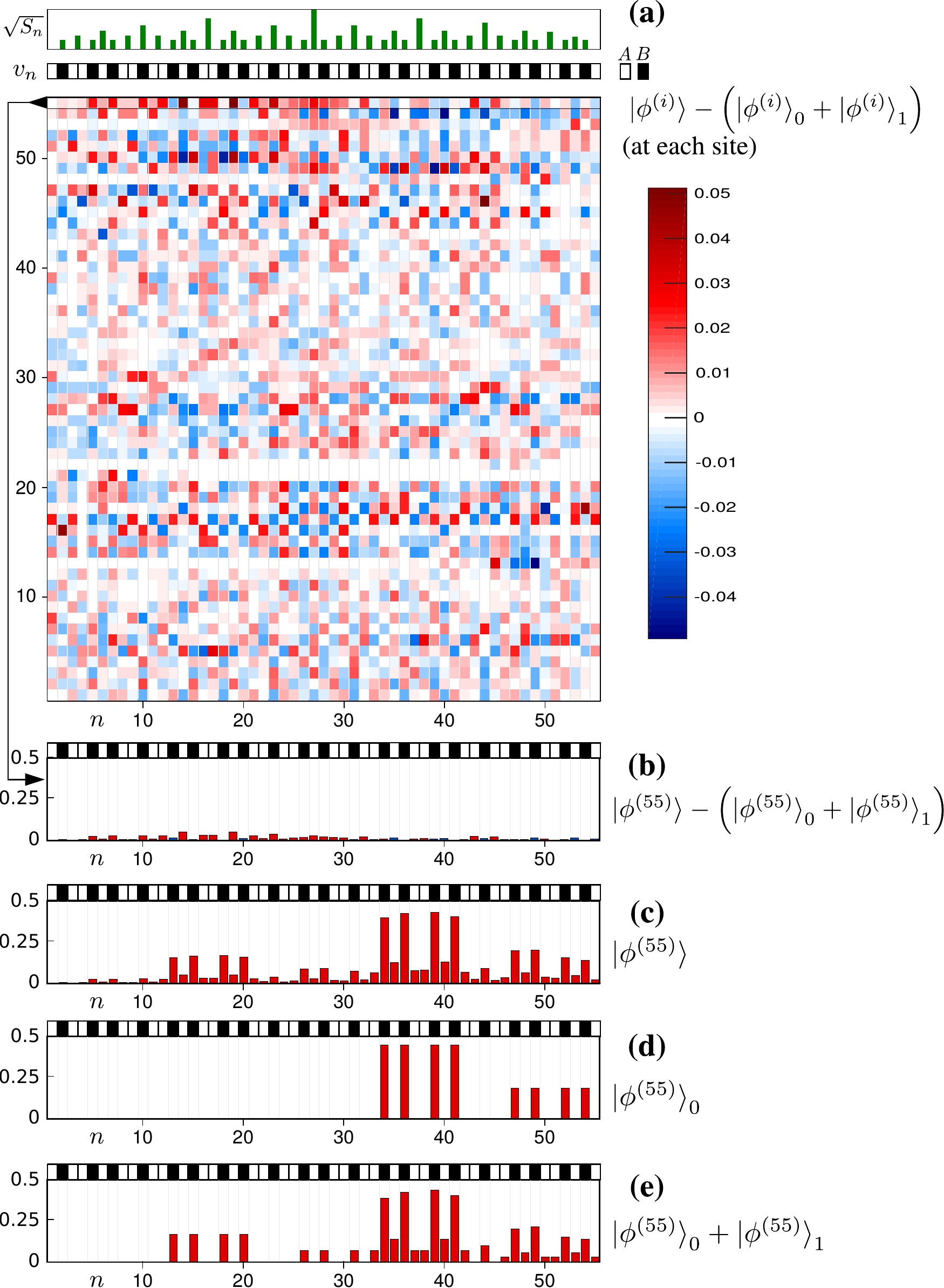}
	\caption{\textbf{(a)} Above: Distribution of axes of local symmetry domains and potential sequence, which is identical to that in \cref{fig:fibo55StateMap} (a), i.e., corresponds to a ninth generation Fibonacci chain. Below: At each site, the map shows the difference between the full eigenstate $\ket{\phi^{(i)}}$ and the sum of the zeroth-order state and the first-order correction at a contrast of $c = 6$. \textbf{(b)} Detailed view on these differences for the uppermost state. The sign of amplitudes is color coded, red for positive and blue for negative values. \textbf{(c)} The amplitudes of the eigenstate $\ket{\phi^{(55)}}$. Note that this particular state does not contain any negative amplitudes. \textbf{(d)} The amplitudes of the zeroth-order state $\ket{\phi^{(55)}}_{0}$. \textbf{(e)} The amplitudes of the zeroth-order state $\ket{\phi^{(55)}}_{0}$ plus that of the first-order state correction $\ket{\phi^{(55)}}_{1}$ (not normalized).}
	\label{fig:fibo55Comparison}
\end{figure}

\begin{figure}[]
	\centering
	\includegraphics[max size={.99\columnwidth}{0.75\textheight}]{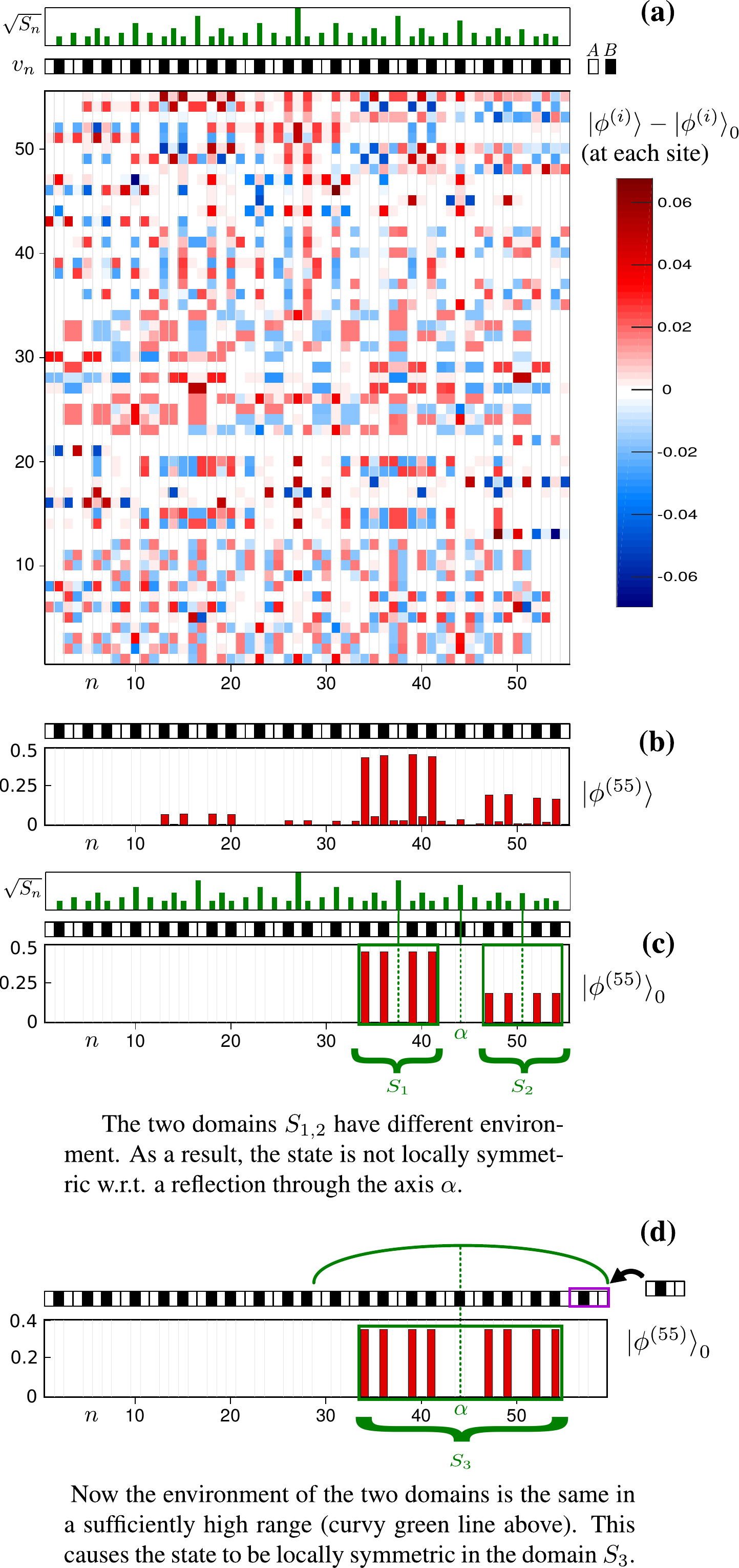}
	\caption{
		\textbf{(a)} Same as in \cref{fig:fibo55Comparison} (a), but now at a contrast $c = 15$ and without the first-order correction $\pState{i}{1}$.
		\textbf{(b)} The uppermost eigenstate $\ket{\phi^{(55)}}$. \textbf{(c)} the zeroth-order state $\ket{\phi^{(55)}}_{0}$. Within $S_{1,2}$ the state is locally symmetric w.r.t. a reflection at the respective centers of these domains (indicated by dotted lines). However, the state is asymmetric w.r.t. a reflection through the axis $\alpha$. \textbf{(d)} 
		The environment of $S_{1,2}$ has been made symmetric by adding the sites $ABAA$ on the right-hand side. As a result, the zeroth-order state $\ket{\phi^{(55)}}_{0}$ (and also, albeit only approximately, the corresponding complete state, though not shown here) is locally symmetric w.r.t. a reflection through $\alpha$.
	}
	\label{fig:fibo55Comparison2}
\end{figure}
\Cref{fig:fibo55Comparison} demonstrates the applicability of degenerate perturbation theory to a $9$-th generation Fibonacci chain [the same as shown in \cref{fig:fibo55StateMap} (a)] at a contrast $c = 6$.
In subfigure (a), at each site the difference
\begin{equation*}
\delta^{(i)} = \pState{i}{} - \left( \pState{i}{0} + \pState{i}{1} \right) ,\; 1 \le i \le N = 55
\end{equation*}
is shown.
Note that the differences $\delta^{(i)}$ are rather small, and in \cref{fig:fibo55Comparison} (b), a detailed picture is given for the uppermost state $\ket{\phi^{(55)}}$. In \cref{fig:fibo55Comparison} (c) and (d), the full state $\pState{55}{}$ and $\pState{55}{0} + \pState{55}{1}$ are shown, respectively. As one can see, already the zeroth-order state matches the fragmentation behavior of the full state quite well, up to the two double resonator modes $A|\overline{B}A\overline{B}|A$ on the left half of the chain. In \cref{fig:fibo55Comparison} (e), we include the first-order correction $\ket{\phi^{(55)}}_{1}$. As one can see, the resulting state $\pState{55}{0} + \pState{55}{1}$
is very close to the full state $\ket{\phi^{(55)}}$ shown in \cref{fig:fibo55Comparison} (c).
Although we have here only shown the $55$-th state (i.e., uppermost) state in detail, the behavior for all other states is similar. This shows that already the first-order state corrections yield very good results.

If one goes to even higher contrast, already the zeroth-order states $\ket{\phi^{(i)}}_{0}$ are sufficient to get a full picture of the fragmentation of a given state. This is demonstrated in \cref{fig:fibo55Comparison2} for a comparatively very high contrast of $c = 20$. Subfigure (a) shows the difference $\ket{\phi^{(i)}} - \ket{\phi^{(i)}}_{0}$
at each site.
The subfigures (b) and (c) show the complete state $\ket{\phi^{(55)}}$ and the zeroth-order state $\ket{\phi^{(55)}}_{0}$, for which the main features (the resonator modes) are visible very well.
Again, this behavior is the same for all other states, indicating that already the zeroth-order states give a good representation of the localization patterns occurring in the full state.
Before we explicitly show the computations for the first three orders in degenerate perturbation theory, let us comment on the connection between the symmetry of the underlying potential sequence and that of the non-negligible amplitudes of a given eigenstate by means of \cref{fig:fibo55Comparison2} (c).
As can be seen, the zeroth-order state is locally parity symmetric individually within the two domains $S_{1,2}$. However, as a whole this state $\ket{\phi^{(55)}}_{0}$ is not locally reflection symmetric w.r.t. an axis denoted by $\alpha$.
As we will outline in the following, the reason for this is that the \emph{environment} of the two domains $S_{1,2}$ is different, where environment includes not only next-neighboring sites but also the ones located further away (we will explain the notion of ``further away'' in more detail below).
In \cref{fig:fibo55Comparison2} (d), we change the environment of the right domain such that it matches that of the first domain up to the first five neighbors. As a result, the zeroth-order state is now symmetric w.r.t. a reflection through the axis $\alpha$.
In the following, we will investigate the connection between local symmetries of the underlying chain and that of the zeroth-order states in more detail. Finally, we will investigate the first-order state corrections and their relation to the fragmentation of eigenstates.

\subsubsection*{Emergence of localization patterns and their locally symmetric character}
\begin{figure*}[t]
	\centering
	\includegraphics[max size={2\columnwidth}{0.9\textheight}]{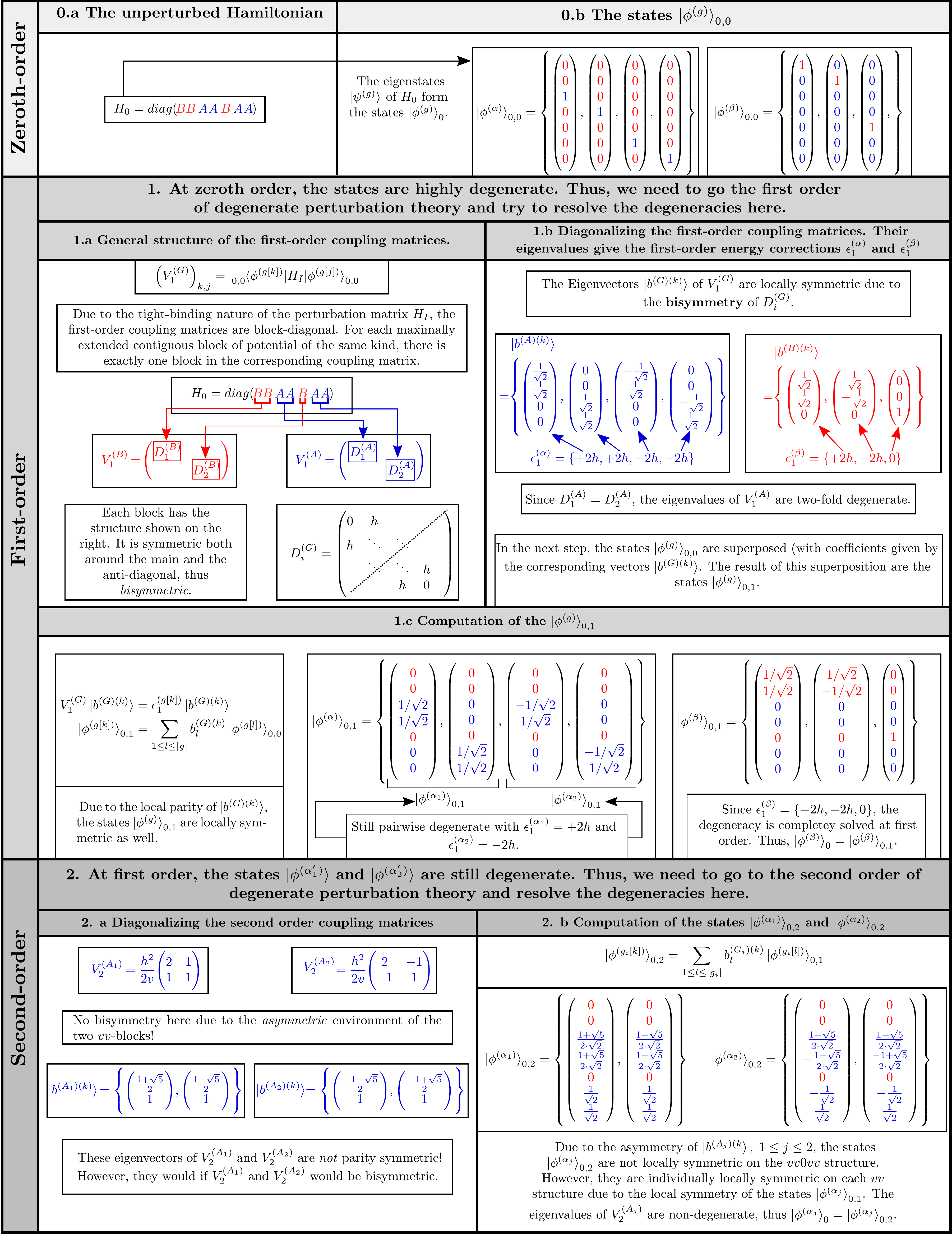}
	\caption{Visualization of the process of finding the zeroth-order states for $H_{0}=diag(B,B,A,A,B,A,A)$.}
	\label{fig:perturbationTheory}
\end{figure*}
We will now show the procedure of finding the zeroth-order states, as can be found e.g. in Refs. \onlinecite{Hirschfelder1974JCP601118DegenerateRSPerturbationTheory,Silverstone1971JCP542325ExplicitSolutionWavefunctionEnergy}.
Since this procedure is quite technical, to help the reader we have visualized the process in a concise form in \cref{fig:perturbationTheory} for the easily traceable case of $H_{0}=diag(B,B,A,A,B,A,A)$.

As stated above, $\pState{g}{0}$ can in general not be determined before its degeneracy is not completely lifted.
At higher orders, the states $\pState{g}{}$ degenerate at zeroth order may split into subsets $\pState{g_{1}}{},\pState{g_{2}}{},\ldots{}$ which are degenerate up to first order, each of which can subsequently split into subsets of states $\pState{g_{1,1}}{},\pState{g_{1,2}}{},\ldots{}$, $\pState{g_{2,1}}{},\pState{g_{2,2}}{},\ldots{}$ which are degenerate up to second order, and so on.
The determination of the zeroth-order states can be done by means of recursively defined auxiliary states\cite{Hirschfelder1974JCP601118DegenerateRSPerturbationTheory,Silverstone1971JCP542325ExplicitSolutionWavefunctionEnergy}
\begin{align}
	\pState{g}{0,0} &= \psiState{g}{}\\
	\pState{g[k]}{0,1} &= \sum\limits_{1 \le l \le |g|} b_{l}^{(G)(k)} \pState{g[l]}{0,0} \label{eq:firstSuperposition}\\
	\pState{g_{i}[k]}{0,2} &= \sum\limits_{1 \le l \le |g_{i}|} b_{l}^{(G_{i})(k)} \pState{g_{i}[l]}{0,1} \label{eq:secondSuperposition}\\
	\pState{g_{i,j}[k]}{0,3} &= \sum\limits_{1 \le l \le |g_{i,j}|} b_{l}^{(G_{i,j})(k)} \pState{g_{i,j}[l]}{0,2} \label{eq:thirdSuperposition}\\
	&\;\;\;\vdots{}
\end{align}
appearing on the left-hand side of the above equations, where $g[k]$ denotes the $k$-th element of the set $g$ and $k$ can run from $1$ to the number of elements $|g|$ within the set. The index $G$ is equal to $A$ if $g \in \alpha$ and equal to $B$ if $g \in \beta$.
Each expansion coefficient $b_{l}^{(S)(k)}$ is the $l$-th component of the vector $\ket{b^{(S)(k)}}, S \in \{G,G_{i},G_{i,j},\ldots{}\}$ defined as
\begin{align*}
	V_{1}^{(G)} \ket{b^{(G)(k)}} &= \en{g[k]}{1} \ket{b^{(G)(k)}} \\
	V_{2}^{(G_{i})} \ket{b^{(G_{i})(k)}} &= \en{g_{i}[k]}{2} \ket{b^{(G_{i})(k)}} \\
	V_{3}^{(G_{i,j})} \ket{b^{(G_{i,j})(k)}} &= \en{g_{i,j}[k]}{3} \ket{b^{(G_{i,j})(k)}} \\
	&\vdots{}
\end{align*}
where the matrices $V_{1}^{(G)}, V_{2}^{(G_{i})},\ldots{}$ are obtained by a recursive process\cite{Hirschfelder1974JCP601118DegenerateRSPerturbationTheory,Silverstone1971JCP542325ExplicitSolutionWavefunctionEnergy}.
Explicitly, for the first three orders they are
\begin{align*}
	\left(V_{1}^{(G)}\right)_{k,j} &= \leftRightBraket{0,0}{\phi^{(g[k])}|H_{I}|\phi^{(g[j])}}{0,0} \\
	\left(V_{2}^{(G_{i})}\right)_{k,j} &= \leftRightBraket{0,1}{\phi^{(g_{i}[k])}|H_{I}R^{(g)}H_{I}|\phi^{(g_{i}[j])}}{0,1} \\
	\left(V_{3}^{(G_{i,j})}\right)_{k,l} &= \leftRightBraket{0,2}{\phi^{(g_{i,j}[k])}|U^{(g_{i})}|\phi^{(g_{i,j}[l])}}{0,2}
\end{align*}
where
\begin{equation*}
	U^{(g_{i})} = H_{I}R^{(g)}\overline{H_{I}}R^{(g)}H_{I} + H_{I} R^{(g)} H_{I} R^{(g_{i})}H_{I}R^{(g)} H_{I}
\end{equation*}
with $\overline{H_{I}} = H_{I} - \en{g_{i}}{1}$ and
\begin{equation*}
R^{(g)} = \sum_{k \notin g} \frac{\psiState{k}{} \bra{\psi^{(k)}}}{\en{g}{0} - \en{k}{0}},\;\; R^{(g_{i})} = \sum\limits_{\substack{k \in g_{j}\\ j \ne i}} \frac{\pState{k}{0,1} \leftBraket{0,1}{\phi^{(k)}}}{\en{g_{i}}{1} - \en{k}{1}} .
\end{equation*}

The above recursive process does the following: At the start, we have $\pState{g}{0,0} = \psiState{g}{}$. These are then superposed according to \cref{eq:firstSuperposition}, obtaining $\pState{g}{0,1}$. Within each degenerate subspace $g_{i}$, these are again superposed according to \cref{eq:secondSuperposition}, obtaining $\pState{g_{i}}{0,2}$. Again, within each degenerate subspace $g_{i,j}$, these are superposed according to \cref{eq:thirdSuperposition}, obtaining $\pState{g_{i,j}}{0,3}$, and so on.
Provided that the degeneracy of a given state $\pState{k}{},\; 1 \le k \le N$ is solved at $n$-th order, the degenerate subspace for this state at orders $l>n$ contains only one state, so that naturally $\pState{k}{0,l} = \pState{k}{0,n}$ and\cite{Hirschfelder1974JCP601118DegenerateRSPerturbationTheory} $\pState{k}{0} = \pState{k}{0,n}$.

In the following, we will prove that the $\pState{i}{0},\; 1 \le i \le N$ simultaneously localize on one or more maximally extended blocks of potentials of the same kind (MEBPS) [Statement 1] and determine on which such blocks a given state can simultaneously localize (Statement 2).
Each MEBPS is the cavity of a resonator, thus giving reason for the localization of states on resonators. Statement 1 also shows that the $\pState{i}{0}$ are locally parity symmetric individually on each MEBPS, and Statements 3 and 4 further deal with longer-range symmetries of the zeroth-order states.
Out of the many possible choices of $\psiState{i}{}$ (due to its high degeneracy), in the following we chose them such that $\psiState{\alpha[k]}{}$ [$\psiState{\beta[k]}{}$] is solely localized on the $k$-th site with potential $A$ [$B$] (counted from the left).

\begin{statement} \label{stat:MEBPS}
	Each state $\pState{i}{0},\; 1 \le i \le N$ simultaneously localizes on one or more maximally extended blocks $A\ldots{}A$ or $B\ldots{}B$ of potentials of the same kind (MEBPS) and is locally parity symmetric on each of these blocks.
\end{statement}
\begin{proof}
	The proof is done in three steps. Firstly, we show that $V_{1}^{(G)},\; G \in \{A,B\}$ is block-diagonal, where each block is related to exactly one MEBPS. Secondly, we show that the eigenvectors $\ket{b^{(G)(k)}},\; 1\le k \le |g|$ of $V_{1}^{(G)}$ are locally parity symmetric and subsequently the $\pState{g}{0,1}$ are locally symmetric on each MEBPS. Thirdly, we show that any higher-order states $\pState{g}{0,n}$ with $n>1$ show this local symmetry as well, and thus the zeroth-order states $\pState{g}{0}$ are locally symmetric as well.
	
	We start by proving the following.
	For the case that $H_{I}$ contains only next-neighbor couplings (as is the case here) the $V_{1}^{(G)}$ become block-diagonal, i.e., can be written as
	\begin{equation} \label{eq:blockDiagonalMatrix}
	V_{1}^{(G)} = 
	\begin{pmatrix}
	D_{1}^{(G)}\\
	& \ddots \\ 
	& & D_{n_{G}}^{(G)}
	\end{pmatrix}
	\end{equation}
	where $n_{G}$ denote the number of blocks occurring in $V_{1}^{(G)}$ and each block
	\begin{equation} \label{eq:tridiagonalCouplingMatrix}
	D_{i}^{(G)} = 
	\begin{pmatrix}
	0 & h & & &\\
	h & \ddots & \ddots & & \\
	& \ddots & \ddots & h \\
	& & h & 0
	\end{pmatrix},\; 1 \le i \le n_{G}
	\end{equation}
	is a tridiagonal Toeplitz matrix.
	To prove that $V_{1}^{(G)}$ is of the above form, we note that by the definition of $V_{1}^{(G)}$ and $H_{I}$, two states $\psiState{j}{}, \psiState{k}{},\; 1 \le j,k \le N$ are coupled to each other by any of the two matrices $V_{1}^{(G)}$ provided that (i) the single sites on which they localize are direct neighbors and (ii) they have the same zeroth-order energy, i.e., they must be localized on states with identical on-site potential.
	If (i) and (ii) are fulfilled for $\psiState{j}{}, \psiState{k}{},\; j \ne k$ with $j,k \in g$ and $g[l] = j_{1},g[m] = j_{2}$, then the corresponding matrix element $(V^{(G)}_{1})_{l,m} = h$ due to the definition of these states.
	As a result, for each MEBPS $A\ldots{}A$ [$B\ldots{}B$] containing $n$ sites, there is one tridiagonal $n\times{}n$ block of the form \cref{eq:tridiagonalCouplingMatrix} present in $V^{(A)}$ [$V^{(B)}$].
	
	We now show that the $\pState{g}{0,1}$ are locally parity symmetric on each MEBPS.
	To this end, we use the fact that the eigenvectors of the block-diagonal matrix $V^{(G)}$ are
	\begin{equation*}
	\ket{b^{(G)(k)}} = 
	\left\{
		\begin{pmatrix}
		\{ \mathbf{w}_{1,G} \} \\
		\mathbf{0}_{d_{2,G}}\\
		\mathbf{0}_{d_{3,G}}\\
		\vdots\\
		\mathbf{0}_{d_{n_{G}},G}
		\end{pmatrix}
	,
		\begin{pmatrix}
		\mathbf{0}_{d_{1,G}}\\
		\{ \mathbf{w}_{2,G} \} \\
		\mathbf{0}_{d_{3,G}}\\
		\vdots\\
		\mathbf{0}_{d_{n_{G}},G}
		\end{pmatrix}
	,\ldots ,
		\begin{pmatrix}
		\mathbf{0}_{d_{1,G}}\\
		\mathbf{0}_{d_{2,G}}\\
		\mathbf{0}_{d_{3,G}}\\
		\vdots\\
		\{ \mathbf{w}_{n_{G},G} \}
		\end{pmatrix}
	\right\}
	\end{equation*}
	with $1 \le k \le |g|$ and where $\mathbf{0}_{d_{i,G}}$ is the $d_{i,G}\times 1$ vector with identical zero entries and $\{\mathbf{w}_{i,G}\}$ denotes the set of $d_{i,G}$ eigenvectors of $D_{i}^{(G)} \in \mathbb{R}^{d_{i,G} \times{} d_{i,G}}$.
	All vectors in a given set $\{\mathbf{w}_{i,G}\}$ have non-vanishing components only on one maximally extended block of potentials of the same kind and are parity-symmetric w.r.t. a reflection through the center of this block.
	The latter is due to the fact that the $D_{i}^{(G)}$ are real and bisymmetric, and the eigenstates of such matrices have definite parity\cite{CollarCentrosymmetricCentroskewMatrices1962} (in the case of degeneracies, the eigenvectors can be chosen accordingly). A matrix is bisymmetric if it is symmetric both around the main and the anti-diagonal.
	Since we have ordered the state $\psiState{g[k]}{},\; 1 \le k \le |g|$ such that it has non-vanishing amplitude on the $k$-th site with potential $G$, one can easily show that each of the $\pState{g}{0,1}$ has definite parity on each MEBPS.
	
	For second order degenerate perturbation theory, the states $\pState{g_{i}}{0,1}$ which are degenerate up to first order are superposed to obtain $\pState{g_{i}}{0,2}$.
	Now, since all states in a given set $\{\mathbf{w}_{j,g} \},\; 1 \le j \le n_{G}$ have distinct eigenvalues, the states $\pState{g_{i}}{0,1}$ are constructed such that for each set $g_{i}$ there is at most one state possessing non-vanishing amplitudes on any given MEBPS.
	Thus, $\pState{g_{i}}{0,2},\ldots{}$ will keep the local parity symmetry, and it is trivial to show that the zeroth-order states $\pState{g}{0}$ are locally parity symmetric on each MEBPS as well.
\end{proof}
Due to its maximal extension, each MEBPS is directly neighbored either by potentials of the other kind on one or on both sides, with the former being the case if the MEBPS forms one edge of the chain. Thus, the $\pState{g}{0}$ are seen to localize on \emph{resonators}.
We now show that a given state $\pState{i}{0},\; 1 \le i \le N$ can only simultaneously localize on resonators fulfilling certain conditions.
\begin{statement}
	A given zeroth-order state $\pState{i}{0},\; 1 \le i \le N$ can simultaneously localize on a set of MEBPS with individual lengths $l_{1},l_{2},\ldots{},l_{n}$ only if the following conditions are met. (i) All the MEBPS must have potentials of the same kind. (ii) There exist integers $1 \le k_{j} \le l_{j},\; 1 \le j \le n$ such that
	\begin{equation} \label{eq:integerConditionResonators}
		\frac{k_{1}}{l_{1} + 1} = \frac{k_{2}}{l_{2} + 1} = \ldots{} = \frac{k_{n}}{l_{n} + 1}.
	\end{equation}
\end{statement}
\begin{proof}
	By definition, the zeroth-order state $\pState{i}{0}$ is formed by superpositions of a subset of the states $\pState{g_{j}}{0,1}$, with $i \in g_{j}$.
	Thus, a necessary condition to allow for the localization on multiple MEBPS $\{M_i\}$ is that among $\ket{\phi^{(g_j)}}_{0,1}$, for each $M_i$ there is one state localized on it.
	By definition, the set $\pState{g_{j}}{0,1}$ contains states with pairwise identical zeroth-order $\en{g_{j}}{0}$ \emph{and} pairwise identical first-order energy corrections $\en{g_{j}}{1}$. The zeroth-order energies are identical if all the MEBPS have the same potential. To see when there is an equality of the of first-order energies, we use the fact that the first-order energy corrections $\en{g}{1}$ can be given analytically.
	The block matrix $D_{i}^{(G)} \in \mathbb{R}^{l_{i} \times{} l_{i}}$ occurring in $V_{1}^{(G)}$ is of tridiagonal Toeplitz form, and its eigenvalues are thus\cite{NoscheseTridiagonalToeplitzmatrices2013} given by
	\begin{equation} \label{eq:toeplitzMatrix}
	\lambda_{k}^{D_{i}^{(G)}} = 2|h| \cos\left(\frac{\pi k}{l_{i} +1}\right),\; k=1,\ldots,l_{i} .
	\end{equation}
	Thus, two blocks $D_{1}^{(G)}, D_{2}^{(G)}$ with size $l_{1}, l_{2}$ only share common eigenvalues
	provided that the integer-equation
	\begin{equation} \label{eq:intersectingEigenvalues}
	\frac{k_{1}}{l_{1} + 1} = \frac{k_{2}}{l_{2} + 1}
	\end{equation}
	is fulfilled for some $1\le k_{1} \le l_{1}$ and $1 \le k_{2} \le l_{2}$. Generalizing the above to the case of $n$ blocks with corresponding length $l_{1},\ldots{},l_{n}$ directly yields \cref{eq:integerConditionResonators}.
\end{proof}
For many combinations of $l_{1} \ne l_{2}$ (especially for small $l_{1,2}$), \cref{eq:intersectingEigenvalues} can not be fulfilled, with the prominent exception of $l_{1,2}$ both being odd numbers.
In this case, there exist states $\pState{i}{0,1}$ which localize on two resonators of \emph{different} kind, and usually this behavior is kept also for $\pState{i}{0}$ as well as the corresponding complete states $\ket{\phi^{(i)}}$.
This is the explanation for the emergence of the two states in \cref{fig:Rudin55StateMap} (a) which are marked by green ellipse.

We now show how the local symmetries of the zeroth-order states can be explained by means of that of the underlying potential sequence.
Due do the complexity of binary tight-binding chains, we only show two explicit cases, but stress that the process can easily be applied to any given chain.
\begin{statement}
	If $H_{0}$ contains one or more of the substructures
	\begin{equation} \label{eq:firstSymmetryCase}
		[\ldots{}]AA\underbracket{BAB}_{S_{1}}AA[\ldots{}]
	\end{equation}
	or
	\begin{equation} \label{eq:firstSymmetryCase2}
	[\ldots{}]AA\underbracket{BAB}_{S_{1}}A
	\end{equation}
	(where $[\ldots{}]$ denotes a possibly larger extension of the chain) then all zeroth-order states $\pState{\beta}{0}$ respect the local symmetry $S_{1}$ on each of these structures.
\end{statement}
\begin{proof}
	We label the sites of the substructure $AABABAA$ from left ($s_{1}$) to right [$s_{7}$ for \cref{eq:firstSymmetryCase} and $s_{6}$ for \cref{eq:firstSymmetryCase2}], where the small $s$ indicates a possible embedding of the corresponding substructure into a greater system.
	Among the $N$ states $\pState{i}{0,1},\; 1 \le i \le N$ of this system, all but the two states $\pState{j_{k}}{0,1},\; 1 \le k \le 2$ with $1 \le j_{k} \le N,\; j_{1} \ne j_{2}$ have vanishing amplitudes on both of the sites $s_{3}$ and $s_{5}$. Moreover, $\pState{j_{1}}{0,1}$ has non-vanishing amplitude only on site $s_{3}$, while $\pState{j_{2}}{0,1}$ has non-vanishing amplitude only on site $s_{5}$.
	We denote the set $\pState{g_{1}}{}$ to contain all states which are degenerate with $\pState{j_{k}}{}$ up to first order.
	As can be shown, $V_{2}^{(G_{1})}$ (just as $V_{1}^{(G)}$) is block-diagonal, and only states that are degenerate up to first order and which are localized on MEBPS which are separated by exactly one site are coupled to each other. Thus, the two states $\pState{j_{1}}{0,1}, \pState{j_{2}}{0,1},\; j_{1},j_{2} \in g_{1}$ are \emph{not} coupled to the other $\pState{g_{1}}{0,1}$ by means of $V_{2}^{(g_{1})}$, but only to each other. If $g_{1}[1] = j_{1}$ and $g_{1}[2] = j_{2}$, then the submatrix
	\begin{equation} \label{eq:V2Matrix}
		\left(V_{2}^{(G_{1})} \right)_{l,m} =
		\frac{h^{2}}{v_{B} - v_{A}}
		\begin{pmatrix}
		2 & 1\\
		1 & 2
		\end{pmatrix},\; 1 \le l,m \le 2
	\end{equation}
	which is real-valued and bisymetric. Its eigenvectors are thus parity-symmetric. As can be easily shown, thus $\pState{j_{k}}{0,2}$ are parity symmetric within $S_{1}$, i.e., respect this domain of local symmetry. The matrix in \cref{eq:V2Matrix} has non-degenerate eigenvalues referring to $\en{j_{k}}{2}$, and thus the two states $\pState{j_{k}}{}$ are no longer degenerate to each other at second order. Since $\pState{j_{k}}{0,1}$ are the only ones out of the $\pState{\beta}{0,1}$ with non-vanishing amplitudes within $S_{1}$, one can easily show that \emph{all} zeroth-order states $\pState{\beta}{0}$ must respect $S_{1}$.
\end{proof}
The above is of relevance for the first and third quasiband from top in \cref{fig:fibo55StateMap} (a).
By means of another example, we indicate the importance of the environment of a domain $S$ such that the zeroth-order states respect it.
\begin{statement}
	If the right edge of $H_{0}$ is given by
	\begin{equation} \label{eq:firstAsymmetryCase}
	[\ldots{}]BAAB\underbracket{ABA}_{S_{1}}
	\end{equation}
	(where $[\ldots{}]$ denotes a possibly larger extension of the chain) then the zeroth-order states $\pState{\alpha}{0}$ do not respect the local symmetry $S_{1}$.
	However, if the right edge of $H_{0}$ is given by
	\begin{equation} \label{eq:firstAsymmetryCaseb}
	[\ldots{}]BAAB\underbracket{ABA}_{S_{1}}B
	\end{equation}
	then all zeroth-order states $\pState{\alpha}{0}$ respect the local symmetry $S_{1}$.
\end{statement}
\begin{proof}
We label the sites of the substructure $BAABABA$ from left ($s_{1}$) to right [$s_{7}$ for the first and $s_{8}$ for the second statement].
Among the $N$ states $\pState{i}{0,1},\; 1 \le i \le N$ of this system, all but the two states $\pState{j_{k}}{0,1},\; 1 \le k \le 2$ with $1 \le j_{k} \le N,\; j_{1} \ne j_{2}$ have vanishing amplitudes on both of the sites $s_{5}$ and $s_{7}$. Moreover, $\pState{j_{1}}{0,1}$ has non-vanishing amplitude only on site $s_{5}$, while $\pState{j_{2}}{0,1}$ has non-vanishing amplitude only on site $s_{7}$.
We denote the set $\pState{g_{1}}{}$ to contain all states which are degenerate with $\pState{j_{k}}{}$ up to first order.
Again, due to the block-diagonal character of $V_{2}^{(g_{1})}$, the two states $\pState{j_{1}}{0,1}, \pState{j_{2}}{0,1},\; j_{1},j_{2} \in g_{1}$ are \emph{not} coupled to the other $\pState{g_{1}}{0,1}$ by means of $V_{2}^{(G_{1})}$, but only to each other. If $g_{1}[1] = j_{1}$ and $g_{1}[2] = j_{2}$, then the submatrices
\begin{equation} \label{eq:V2MatrixAsymmetric}
\left(V_{2}^{(G_{1})} \right)_{l,m} =
\frac{h^{2}}{v_{B} - v_{A}}
\begin{pmatrix}
2 & 1\\
1 & 1
\end{pmatrix},\; 1 \le l,m \le 2
\end{equation}
for $[\ldots{}]BAABABA$ and
\begin{equation} \label{eq:V2MatrixAsymmetricb}
\left(V_{2}^{(G_{1})} \right)_{l,m} =
\frac{h^{2}}{v_{B} - v_{A}}
\begin{pmatrix}
2 & 1\\
1 & 2
\end{pmatrix},\; 1 \le l,m \le 2
\end{equation}
for $[\ldots{}]BAABABAB$.
Both \cref{eq:V2MatrixAsymmetric,eq:V2MatrixAsymmetricb} are real-valued, but the former is not bisymmetric, while the latter is.
As can be easily shown, for the first case, the $\pState{j_{k}}{0,2}$ are also not parity symmetric within $S_{1}$, i.e., do not respect this domain of local symmetry. The matrix in \cref{eq:V2MatrixAsymmetric} has non-degenerate eigenvalues referring to $\en{j_{k}}{2}$, and thus the two states $\pState{j_{k}}{}$ are no longer degenerate to each other at second order. Since $\pState{j_{k}}{0,1}$ are the only ones out of the $\pState{\alpha}{0,1}$ with non-vanishing amplitudes within $S_{1}$, one can easily show that no zeroth-order state $\pState{\alpha}{0}$ respects $S_{1}$.
For the second case, the line of argumentation essentially is the same with the difference that, due to the bisymmetry of \cref{eq:V2MatrixAsymmetricb}, the $\pState{j_{k}}{0,2}$ are parity symmetric within $S_{1}$, and thus all $\pState{\alpha}{0}$ respect $S_{1}$.
\end{proof}
The reason for the non-bisymmetry of \cref{eq:V2MatrixAsymmetric} is the different \emph{environment} of $s_{5}$ and $s_{7}$. In this particular case, the environment is made up by the next-neighboring sites, but for higher orders it comprises many more sites left and right to the given domain. The fact that $\pState{55}{0}$ in \cref{fig:fibo55Comparison2} (c) is not locally symmetric w.r.t. a reflection through $\alpha$ is due to the fact that the environment of $S_{1,2}$ is not symmetric w.r.t. a reflection through $\alpha$ in a sufficiently large radius, while in \cref{fig:fibo55Comparison2} (d) it is, so that $\pState{55}{0}$ is locally symmetric w.r.t. a reflection through $\alpha$.

\subsubsection*{First-order state corrections and eigenstate fragmentation}
In the above, we have seen how the correct zeroth-order states are related to the local symmetries of the underlying potential. In particular, we have argued that each of the $\pState{i}{0}$ is fragmented, since it has non-vanishing amplitudes only on one kind of site.
We have further seen that, already at contrast $c= 6$, $\pState{i}{0} + \pState{i}{1} \approx \pState{i}{}$.
In the following we show that, in general, the $\pState{i}{0} + \pState{i}{1}$ are fragmented as well.

Contrary to the non-degenerate case, where the first order state corrections are given by
\begin{equation*}
	\pState{i}{1} = \sum_{j \ne i}\frac{\pState{j}{0}\leftBraket{0}{\phi^{(j)}}}{\en{i}{0} - \en{j}{0}} H_{I} \pState{i}{0},
\end{equation*}
the corresponding expression in degenerate perturbation theory depends on the order in which the degeneracy of $\pState{i}{}$ is completely resolved. A full, recursive expression for $\pState{i}{1}$ can be found in Ref. \onlinecite{Silverstone1981PRA231645PracticalRecursiveSolutionDegenerate}. In this context, we only need the easily provable fact that
\begin{equation}
	\leftRightBraket{}{\psi^{(\bar{g}[j])}|\phi^{(g[k])}}{1} = \frac{\braket{\psi^{(\bar{g}[j])}|H_{I}| \phi^{(g[k])}}_{0}}{\en{g[k]}{0} - \en{\bar{g}[j]}{0}}
\end{equation}
where $\bar{g}$ denotes the set of sites which are not elements of $g$. In other words, if $\pState{i}{0}$ ``lives'' on, say, sites with potential $A$, then $\pState{i}{1}$ will have non-vanishing amplitudes only on directly \emph{neighboring} $B$ sites, but not on those further away.
As a result, if $\pState{i}{0}$ has non-vanishing amplitudes on a small number of sites (which we have observed for Fibonacci, Thue-Morse and Rudin-Shapiro chains), then $\pState{i}{0} + \pState{i}{1}$ is fragmented.

\section{Discrete Energy-localization theorem and approximation of eigenvalues by sub-Hamiltonians}
\label{appendix:submatrixTheorem}
We here extend a theorem of Ref. \onlinecite{Filoche2012PNASU10914761UniversalMechanismAndersonWeak}, connecting the localization of a state to its eigenenergy, to discrete Hamiltonians:
\begin{theorem}
	The following equation holds
	\begin{equation} \label{eq:submatrixJustification}
	\frac{\Vert \ket{\phi} \Vert_{\partial \overline{D}}}{\Vert \ket{\phi} \Vert_{D}} \ge \min_{\epsilon_{k}} \frac{| \epsilon - \epsilon_{k}|}{|h|}
	\end{equation}	
	where $\ket{\phi}$ is an eigenvector of $H$ with energy $\epsilon$ and $\epsilon_{k}$ are eigenvalues of $H$ restricted to the domain $D$ which is a simply connected subdomain of the whole system. $\Vert \ket{\phi} \Vert_{D}$ is the norm of $\ket{\phi}$ on $D$ and $\Vert \ket{\phi} \Vert_{\partial \overline{D}}$ the norm of $\ket{\phi}$ on next-neighbors of $D$.
\end{theorem}
\begin{proof}
	If $D$ contains $N_{D}$ sites, define the $N_{D} \times N_{D}$ matrix $H_{D}$ constructed from the corresponding matrix elements of the complete Hamiltonian $H$.
	In other words, $H_{D}$ is the restriction of $H$ onto $D$.
 	Similarly, we further define $\ket{i}$ as the $N_{D} \times 1$ vector constructed from the full eigenvector $\ket{\phi}$ by taking the interior elements of $D$. If we now let $H_{D}$ act on $\ket{i}$, one can easily show that
	\begin{equation} \label{eq:localHamiltonian}
		H_{D} \ket{i} = \epsilon \ket{i} - h\ket{\partial \phi}.
	\end{equation}
	where $\epsilon$ is the eigenvalue of the complete state $\ket{\phi}$.
	Here, $h$ denotes the next-neighboring hopping of $H$ (as defined in \cref{eq:tridiagonalHamiltonian}) and $\ket{\partial \phi}$ denotes a $N_{D} \times 1$ vector with zeros everywhere but on the first and last entry. These two non-vanishing entries are constructed by taking the corresponding two elements of $\ket{\phi}$ within $\partial \bar{D}$. If $N_{D} = 1$, then we define the only entry of $\ket{\partial \phi}$ as the sum of the two amplitudes of $\ket{\phi}$ within $\partial \bar{D}$.
	
	To make the notation introduced above more explicit, let us assume that
	\begin{equation}
		H = \begin{pmatrix}
		v_{1} & h & 0 & 0 \\
		h & v_{2} & h & 0 \\
		0 & h & v_{3} & h \\
		0 & 0 & h & v_{4}
		\end{pmatrix},\;
		\ket{\phi} = \begin{pmatrix}
		a \\ b \\ c \\ d
		\end{pmatrix}.
	\end{equation}
	If $D$ would denote the central two sites, then $\ket{i} = (b,c)^{T}$ and $\ket{\partial \phi} = (a,d)^{T}$.
	
	\cref{eq:localHamiltonian} can be interpreted as follows: Provided that $\ket{\phi}$ is identically zero on the next-neighboring sites of $D$, $\ket{i}$ would be an eigenstate to $H_{D}$. However, $\ket{\phi}$ usually has \emph{non-vanishing} amplitudes on sites neighboring to $D$, and thus $\ket{\partial \phi} \ne 0$. Thus, this correction must be included in \cref{eq:localHamiltonian}.	
	
	We now proceed with our proof of \cref{eq:submatrixJustification}. Multiplying from the left with $\bra{\phi_{k}}$, i.e., the $k$-th eigenstate of $H_{D}$, we get
	\begin{equation}
	h\cdot \braket{\phi_{k}|\partial \phi}  = (\epsilon - \epsilon_{k}) \cdot  \braket{\phi_{k}|i}.
	\end{equation}
	Multiplying this expression by its complex conjugate, summing over $k$ and taking the square root of the result, we get
	\begin{equation} \label{eq:norm}
	|h| \left(\sum_{k} | \braket{\phi_{k}|\partial \phi}|^2\right)^{1/2} =  \left(\sum_{k} (\epsilon - \epsilon_{k})^2 \cdot |\braket{\phi_{k}|i} |^2 \right)^{1/2} .
	\end{equation}
	Since the $\ket{\phi_{k}}$ are a complete orthonormal basis set, the left-hand side can be simplified by using the definition of the norm, getting
	\begin{equation}
		|h| \left(\sum_{k} | \braket{\phi_{k}|\partial \phi}|^2\right)^{1/2} = |h| \Vert \ket{\phi} \Vert_{\partial \overline{D}}.
	\end{equation}
	The sum on the right-hand side can be estimated as
	\begin{equation} \label{eq:estimation}
	\sum_{k} (\epsilon - \epsilon_{k})^2 \cdot \Vert\braket{\phi_{k}|i} \Vert^2 \ge \min_{\epsilon_{k}} (\epsilon - \epsilon_{k})^{2} \cdot \sum_{k'} \Vert\braket{\phi_{k'}|i} \Vert^2.
	\end{equation}
	Again, due to the definition of the norm, we can thus write \cref{eq:norm} as
	\begin{equation} \label{eq:norm2}
	|h| \Vert \ket{\phi} \Vert_{\partial \overline{D}} \ge  \min_{\epsilon_{k}} |\epsilon - \epsilon_{k}| \Vert \ket{\phi} \Vert_{D}.
	\end{equation}
	which directly yields \cref{eq:submatrixJustification}.
\end{proof}
Roughly speaking, the theorem states the following. Assume that an eigenstate $\ket{\phi}$ has a high integrated density on some domain $D$, with low amplitudes on the next-neighboring sites left and right of the domain. Then, the energy $\epsilon$ of this eigenstate is approximately equal to the energy of one of the eigenstates $\ket{\phi_{k}}$ of the local Hamiltonian $H_{D}$.
If $D$ is a resonator and $\ket{\phi}$ represents an LRM of $H_{D}$ within $D$ and suitably small amplitudes on next-neighboring sites of $D$, then $\epsilon \approx \epsilon_{i}$, where $\epsilon_{i}$ is the energy of the LRM.

\section{Comments on the application to longer chains} \label{appendix:longerChainsComments}
\begin{figure}[t]
	\centering
	\includegraphics[max size={.99\columnwidth}{0.6\textheight}]{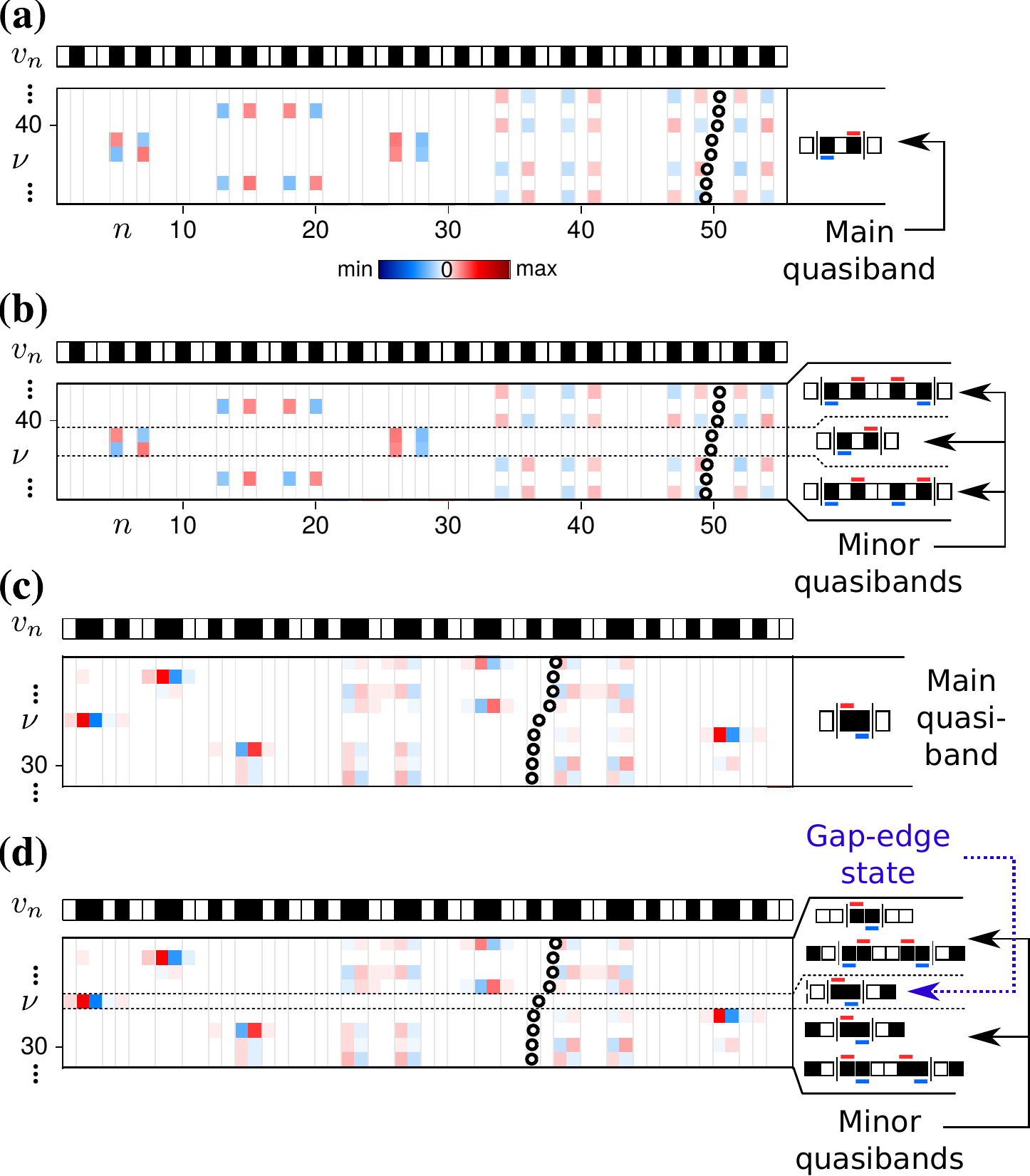}
	\caption{\textbf{(a)} Shown is the third quasiband from top for the $9$-th generation Fibonacci chain at contrast $c = 3$. \textbf{b} The three minor quasibands and their respective LRMs.
	\textbf{(c)} The third quasiband from top for a $L = 55$ sites truncated Thue-Morse chain at contrast $c = 3$. \textbf{d} The two minor quasibands as well as the gap-edge state and their respective LRMs.}
	\label{fig:fibo55StateMapSmall}
\end{figure}
We now comment on how the treatment of longer chains or the investigation of the subband structure can be pursued.
To this end, the core element of our approach, the analysis of states in terms of their constituting LRMs needs to be slightly changed by extending the class of resonators taken into account.
The process of finding the constituting LRMs of a given state $\ket{\phi}$ with energy $\epsilon$ is then as follows.
Starting from a domain $D$ exclusively containing sites with very high amplitudes, one forms a simply connected domain $D'$ by the union of $D$ and its surrounding sites (not limited to next-neighbors) such that $\ket{\phi}$ has very low amplitude on next-neighbors of $D'$.
Then, [guaranteed by \cref{eq:submatrixJustification}], one eigenstate of the Hamiltonian $H_{D'}$ has nearly the same energy $\epsilon \approx \epsilon^{i}$ and is (up to normalization), within $D'$, nearly equal to $\ket{\phi}$ and thus forms an LRM.
As the maximum deviation between $\epsilon^{i}$ and $\epsilon$ is bounded by means of \cref{eq:submatrixJustification} and generally becomes smaller for larger $D'$, its size should thus be chosen large enough to achieve the accuracy needed for an explanation of the sub-quasibands and gap-edge states present, but as small as possible in order not to lose the local character of the treatment.
If the LRM obtained by the above process does not explain all fragments of $\ket{\phi}$, then one needs to repeat it for each of the remaining fragments until all constituting LRMs of $\ket{\phi}$ are found.

We now exemplify in \cref{fig:fibo55StateMapSmall} some possible results of such a deeper analysis.
Subfigure (a) shows the third quasiband from top of the $9$-th generation Fibonacci chain [the one shown in \cref{fig:fibo55StateMap} (a)], but now at a lower contrast of $c = 3$.
At this contrast, the energetical substructure of the band becomes apparent, denoted by the two dashed lines in \cref{fig:fibo55StateMapSmall} (b). There are three minor quasibands, comprising the three uppermost, the two central and the three lowermost eigenstates within this quasiband.
The above process then yields the LRMs shown on the right-hand side of this subfigure.
Another example is demonstrated in \cref{fig:fibo55StateMapSmall} (c) and (d), showing the third quasiband from top for a truncated $L = 55$ site Thue-Morse chain [as shown in \cref{fig:thue55StateMap} (a)] at contrast $c = 3$. Here, the main quasiband features the resonator mode $A|\overline{B}\underline{B}|A$, but again features a substructure as shown in subfigure (d). Each minor quasibands is made up of two nearly degenerate LRMs, with the underlying resonators having resonator walls each consisting of \emph{two} sites.
The state in-between these minor bands consists of the edge-LRM \textbrokenbar{}$A|\overline{B}\underline{B}|AB$, where the \textbrokenbar{} indicates the edge of the chain.


%

\end{document}